\documentclass[11pt]{article}
\usepackage[margin=1.3in]{geometry}
\usepackage{amssymb,amsthm,amsmath,amsfonts}
\usepackage{latexsym}
\usepackage[english]{babel}
\usepackage{tikz}

\usetikzlibrary{shapes.geometric,arrows,fit,matrix,positioning}
\tikzset
{
    treenode/.style = {circle, draw=black, align=center, 
                          minimum size=1cm, anchor=center},
}
\usepackage{verbatim}
\usetikzlibrary{arrows,shapes}

\theoremstyle{plain}
\newtheorem{theorem}{Theorem}
\newtheorem{lemma}[theorem]{Lemma}
\newtheorem{proposition}[theorem]{Proposition}
\newtheorem{corollary}[theorem]{Corollary}

\theoremstyle{definition}
\newtheorem{definition}{Definition}
\newtheorem{example}{Example}
\newtheorem{observation}{Observation}
\newtheorem{remark}{Remark}

\newtheorem{conjecture}{Conjecture}
\newtheorem{property}{Property}

\newcommand{\N}{\ensuremath\mathbb N}

\newcommand{\R}{\mathbb R}
\newcommand{\Ru}{R}

\DeclareMathOperator*{\argmax}{argmax}
\DeclareMathOperator*{\argmin}{argmin}
\def\vec{\boldsymbol}

\newcommand{\C}{\boldsymbol C}
\newcommand{\re}{\boldsymbol r}
\newcommand{\ut}{\boldsymbol{u}}

\newcommand{\p}{{\ensuremath{\widetilde p}}}
\newcommand{\n}{[n]}

\renewcommand{\a}{\ensuremath\vec a}
\newcommand{\om}{\ensuremath{\vec \omega}}
\newcommand{\sig}{\ensuremath{\vec \sigma}}
\newcommand{\x}{\vec x}
\newcommand{\ve}{\vec v}
\newcommand{\length}{{\rm length}}
\newcommand{\red}{{\rm red}}
\newcommand{\cyc}{{\rm cyc}}

\newcommand{\bs}[1]{\ensuremath{\boldsymbol{#1}}}

\newcommand{\zs}{{\rm zs}}
\newcommand{\si}{{\rm si}}

\newcommand{\G}{\mathcal G}
\newcommand{\Lu}{{\ensuremath{\rm Left}}}
\newcommand{\Ri}{{\ensuremath{\rm Right}}}
\newcommand{\Le}{\ensuremath{\mathcal L}}
\newcommand{\Rig}{\ensuremath{\mathcal R}}
\newcommand{\Pre}{\ensuremath{\mathcal P}}
\newcommand{\Ne}{\ensuremath{\mathcal N}}
\renewcommand{\emptyset}{\varnothing}
\newcommand{\s}{\mathcal S}
\newcommand{\A}{\mathcal A}
\newcommand{\T}{{\rm T}}
\renewcommand{\geq}{\geqslant}
\renewcommand{\leq}{\leqslant}
\renewcommand{\ge}{\geqslant}

\newcount\Comments  % 0 suppresses notes to selves in text
\definecolor{darkgreen}{rgb}{0,0.6,0}

\newcommand{\kibitz}[2]{\ifnum\Comments=1{\textcolor{#1}{{#2}}}\fi}
\newcommand{\rmr}[1]{\kibitz{red}{[RM:#1]}}
\newcommand{\uln}[1]{\kibitz{blue}{[UL:#1]}}

\usepackage{enumerate}

\usepackage[colorlinks]{hyperref}
\usepackage[T1]{fontenc}
\usepackage[utf8]{inputenc}
\usepackage{authblk}
\usepackage{xcolor}

\title{Cumulative Games: \\ \Large{Who is the current player?}\vspace{1 cm}}
 \author[1]{Urban Larsson \thanks{urban031@gmail.com, partially supported by Aly-Kaufman fellowship.}}
 \author[2]{Reshef Meir \thanks{reshefm@ie.technion.ac.il}}
 \author[1]{Yair Zick\thanks{zick@comp.nus.edu.sg}}
 \affil[1]{National University of Singapore, Singapore}
 \affil[2]{Technion--Israel Institute of Technology, Haifa, Israel}

\newcommand{\Z}{\mathbb Z }

\renewcommand{\ge}{\geqslant}

\def\vec{\boldsymbol}

\begin{document}
	\pgfdeclarelayer{background}
\pgfsetlayers{background,main}

\maketitle
\begin{abstract}
Combinatorial Game Theory (CGT) is a branch of Game Theory that has developed largely independently of Economic Game Theory (EGT), and is concerned with deep mathematical properties of two-player zero-sum games recursively defined over various combinatorial structures. The aim of this work is to lay the foundations for bridging the conceptual and technical gaps between CGT and EGT, here interpreted as multiplayer Extensive Form Games, so that they can be treated within a unified framework. More specifically, we introduce a class of 
$n$-player, general-sum games, called {\sc Cumulative Games}, which can be analyzed using tools from both CGT and EGT. We show how two of the most fundamental definitions of CGT, the outcome function and the disjunctive sum operator, naturally extend to the class of {\sc Cumulative Games}. The outcome function allows for efficient equilibrium computation under certain restrictions, while the disjunctive sum operator lets us define a partial order over games according to the advantage that a given player has. Finally, we show that any Extensive Form Game can be written as a {\sc Cumulative Game}.
\end{abstract}

\section{Introduction}
%\rmr{An alternative opening:}
In the famous game of {\sc Nim}, there are several heaps of pebbles on a table (Figure~\ref{fig:NIM}). Two players take turns, where a player is allowed to remove any positive number of pebbles from a single heap, and the player who picks up the last pebble wins the game.\\ %\rmr{We should add a figure of Nim example}
\begin{figure}[ht!]
    \centering
   \includegraphics[width = 0.7 \textwidth]{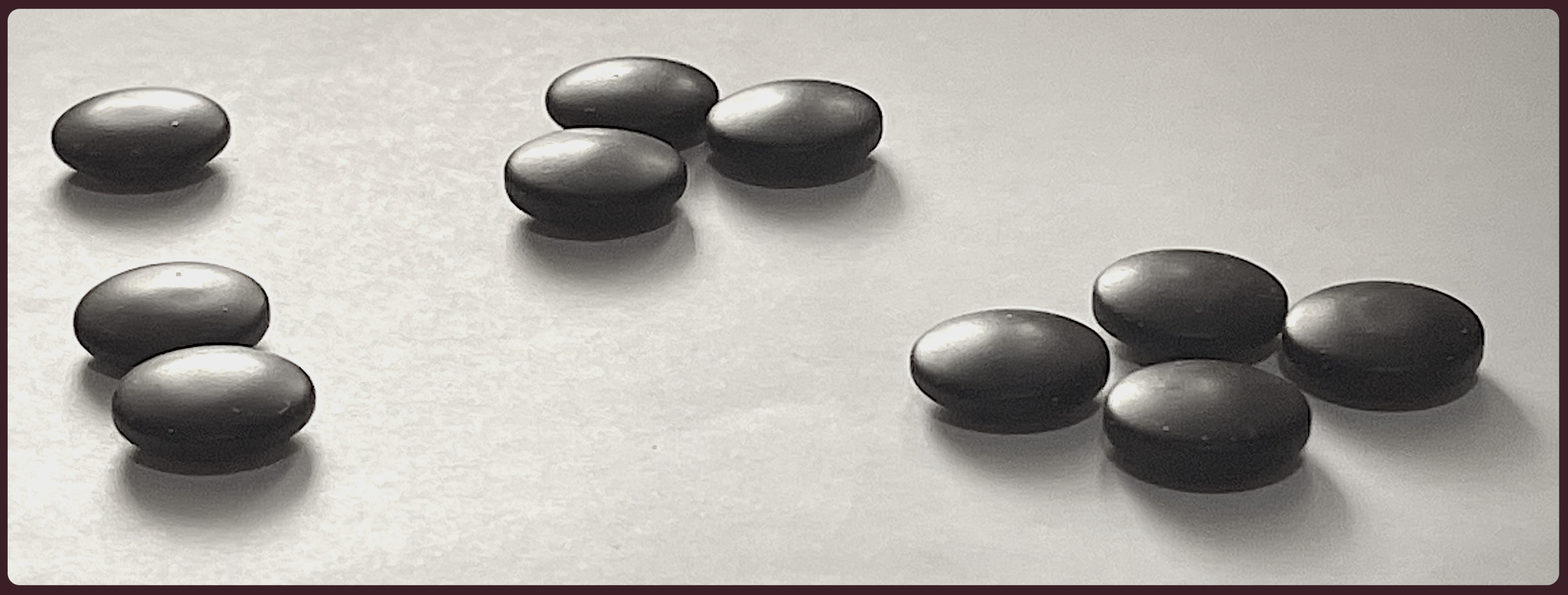}
    \caption{A game of {\sc Nim}.}
    \label{fig:NIM}
\end{figure}

While {\sc Nim} and similar recreational games have been played by people for thousands of years~\cite{CSNS, Jo}, it is in modern times that such games have been analyzed mathematically withing the framework of \emph{Game Theory}, starting with the early works of Bouton~\cite{B}, Zeremelo~\cite{Ze} and von-Neumann~\cite{N}.  
In fact, these early works have inspired at least two branches of Game Theory, that have since developed in different directions and apply different mathematical tools. 

The `main branch' of Game Theory, which we will refer to as \emph{Economic Game Theory} (EGT) uses games as an abstraction for interaction of self-interested parties. Games like {\sc Nim} are viewed as a special case of broader \emph{Extensive Form Games} that describe full information, turn-based interactions between any number of players with arbitrary utility functions. EGT then applies various notions of \emph{equilibrium} (for example Pure Subgame Perfect Equilibrium) that are derived from certain assumptions on players' behavior and rationality. While each game has its own equilibrium outcomes, the focus in EGT literature is typically not on the solution of a particular game, but rather on general properties and techniques,  for example bounds on the efficiency and fairness of equilibrium outcomes, and (more recently) the complexity of algorithms for computing such equilibria~\cite{NRTV}.

In contrast, \emph{Combinatorial Game Theory} (CGT) typically restricts itself to recursively defined  two-player zero-sum games, and focuses on describing the mathematical structures of the winning strategies in specific games and reveal the combinatorial patterns that arise. For example, in {\sc Nim} it is possible to find the winning strategy by considering the binary representation of heap sizes~\cite{B}. Moreover, in CGT we typically analyze whole classes of games together (such as {\sc Nim} with different heap number and sizes), and consider relations of individual positions.  In particular, an important tool in CGT is the \emph{disjunctive sum} operator applied to pairs of games, which intuitively induces a new game where a player can move in one game of her choice in each turn~\cite{BCG}. Conveniently, the disjunctive sum of two {\sc Nim} games is also a {\sc Nim} game, but we can also add up seemingly unrelated games. This special operator can be used to define a partial order over all combinatorial games  and expose unexpected relations, reduction theorems and equivalence classes, sometimes with ``game values'' as simplest representatives. The complexity of such game values grow super exponentially with the rank of its game tree \cite{Su}.

%\rmr{--}
\if 0
Combinatorial Game Theory (CGT) and Economic Game Theory (EGT) have developed almost independently, with a joint origin in Zermelo's backward induction argument \cite{Ze}. CGT is concerned with deep mathematical properties of deterministic two-player zero-sum games that are defined over various combinatorial structures, with key concepts such as ruleset, normal-play, outcome functions and disjunctive sum play. CGT takes inspiration from recreational play traditional board games with no dice and no hidden cards, such as {\sc chess}, {\sc go}, {\sc checkers} and more. On the other hand, EGT revolves around $n$-player simultaneous- or sequential-games in normal- or extensive form, with randomized strategies on general-sum utilities, and often with incomplete information. EGT takes inspiration from the market, real world economic situations, voting and more. In case of complete information, a main analytical tool of EGT concerns subgame-perfect equilibria (SPE), while in CGT one often relies, without mention, on straightformward variations of Zermelo's result (which in essence is a pure SPE).

\uln{Reshef, please fill in correct this part from your point of view. I see now what you wrote: "1 paragraph on EGT (any number of players, general utility functions, main analytical tool is subgame-perfect equilibrium applied to a single game); and one on CGT (2 players, 0-sum, main analytical tool is disjunctive sum applied to classes of games)."}\rmr{I did}
\fi

The aim of this work is to lay foundations to bridging some conceptual and technical gaps between CGT and EGT, so they can be treated within a unified framework. 
More specifically, we introduce a class of $n$-player, general-sum games, called {\sc Cumulative Games}, that can be analyzed by both  CGT and EGT tools.  %\rmr{not sure we should refer to theorems here}\uln{OK, we can have a contribution section where we state and refer to all results.}
{\sc Cumulative Games} are broad enough to model any Extensive Form Game, yet we show how two of the most fundamental definitions of CGT---the outcome function, and the disjunctive sum operator---naturally extend to this class. Moreover, we define the subclass {\sc Heap Dynamic} that permits a tractable analysis of {\em PSPE-outcomes}, which generalizes the classical dynamic programming {\em tabular model} (a.k.a. ``shift register'' model  \cite{Go}) for {\sc Subtraction Games}. %The defined class of games is very
% \rmr{I think that some of this paragraph is too detailed (e.g. the table approach and the shift-register). The other parts should be more connected to the notions we already mentioned from EGT or CGT. e.g. emphasize that we can extend the partial order over games from 2-player zero-sum (was it full order there?) to other classes.}\uln{You are right that it is too specific at an early stage. But before it was completely missing. So something between is preferable. The "table approach" guided me all the way through this work. But it is not visible now. Section~9 used to be in the beginning a long time ago. Should we have a single table here in the beginning and then say that we generalize this approach in the setting of {\sc Cumulative Game}s?}\rmr{I suggest that in this subsection we focus only on *what* we do, not *how*, so remove this paragraph from here. In the end of Sec.1.1 after we informally presented some example games, we can say something about how we solve them (e.g. the game A+C), maybe with an example table as you said.  }

We relax some of the standard combinatorial games' axioms, but not all. We allow for $n\ge 2$ players with general utilities, thereby substantially expanding beyond traditional combinatorial games, for example introducing ``tragedy of the common'' situations to CGT. But we do not yet allow chance moves, hidden information, simultaneous moves and mixed strategies. %The  examples in Subsection~\ref{sec:intrex} are motivational to our approach. Then, we continue the more formal introduction in Subsection~\ref{sec:introduction}.

%Here we will de-emphasize the usual appeal of patterns and combinations in the outcomes/values of games, and instead lay a foundation for a bridge to mainstream Game Theory, i.e. Economic Game Theory (EGT), via multiplayer Extensive Form Games. 
The main contribution of this paper is conceptual, and takes the form of a discussion, which makes it somewhat different from many theoretical papers. Most of the space is dedicated to exploring various definitions that enable us to reconcile diverging concepts and modelling assumptions in CGT and EGT, and now under one umbrella. In particular, we bridge the conventions that finite games in EGT are rooted trees, that require a particular initial state, and where every node has a coupled player, whereas in CGT states are not a-priori coupled with a current player, and there is usually infinitely many possible initial states.

%\subsection{CGT Background}\label{sec:introduction}
%\uln{I am now questioning if this subsection is  needed. I try by removing the header, and we can see which parts will still be good to keep, as a continuation of the previous paragraph. I think it is misplaced here.}
%Since the introduction of Chess-like games and their recursive solution by von Neumann, noncooperative Game Theory has developed in two almost independent directions. Mainstream work in economics started to consider games with multiple players and general utility functions (EGT), and generalized game values by considering subgame perfect equilibrium. 
%\rmr{We already said that in the opening so either we want to add more background on EGT if you think it is needed (e.g. we can mention other directions in EGT like partial information or simultaneous moves), or we remove this repetition and this subsection becomes 'CGT Background'.}

 Another interesting feature of CGT is that it deals with games that people can actually (and often do) play for recreation and/or professional competition.
Many combinatorial games use the \emph{normal play} convention, meaning that the player who can no longer move loses (as in {\sc Nim}). 
%\uln{Should we remove this paragraph? It seems a bit misplaced.}
In fact, both the Milnor theory \cite{M} (on positional games), and the breakthrough via normal play theory  \cite{Co, BCG}, took inspiration from endgame studies in the classical eastern game of {\sc Go}, and CGT disjunctive sum theory was born in these ways.\footnote{We now know that Milnor's ``positional games'' behave as normal play games without  infinitesimals; see \cite{LNS} for a discussion.} In this context one should mention that Berlekamp has combined CGT with EGT by playing out endgames  of combinatorial games together with coupon stacks, to estimate their \emph{temperatures} (see \cite{Be, BW} for more on these topics),  but this topic leads towards a different direction than this work.

%As we said in the opening, combinatorial Game Theory (CGT)  focuses on two-player, zero-sum games, with complete and perfect information, as abstract mathematical entities. 

A game is usually defined via its \emph{ruleset}, which specifies the allowed moves but not the opening position or player turns, e.g. {\sc Hackenbush}, {\sc Domineering} and so on \cite{Co, BCG}; due to its central position in CGT, we use caps fonts for rulsets and ruleset families. Typical results in the CGT literature deal with understanding both winning strategies (as is also done in EGT) and the patterns of the positions in which one or other player may have an advantage (a question that is usually meaningless in EGT, where player roles are fixed). 
 Classes of games are studied by defining various operations such as game addition and game comparison, which take into account games' underlying combinatorial structures together with well-known reduction theorems \cite{Si}.  
 
% Such {\em rulesets} are central to CGT, so we use caps to highlight their names.
%\rmr{we already said that so removing:}\uln{In both strands of the literature much work is dedicated to understanding the outcome of a game when players play rationally, but often using different mathematical tools and even more different terminology. In this paper our purpose is to start building bridges between CGT \cite{Si} and ``economic'' (or mainstream) Game Theory---in particular Extensive Form Games and subgame perfect equilibrium.
%We study a multiplayer, general-sum extension of the classical {\em Subtraction Games} \cite{BCG} (a variation of Nim), called \emph{{\sc Cumulative Game}s}.}

%\uln{This seems like a repitition, yes.}
\subsection{A generalized Cumulative Subtraction: initial examples}\label{sec:intrex}
%\paragraph{}

The ruleset family {\sc Cumulative Subtraction} (CS) is introduced in Stewart's Ph.D. Thesis \cite{St} in a zero-sum setting (also published in \cite{St2}),  with more recent work in \cite{CLMW}. 
Our main conceptual contribution in this paper, {\sc {\sc Cumulative Game}s}, will be based on ideas that further generalize CS to a general-sum setting.  
%Apart from the winning condition, our games are all based on this simple idea. 
To the authors best knowledge, general-sum variations of combinatiorial games have not yet been studied in the literature.\footnote{Since the first version of this manuscript, a study inspired by this work has appeared \cite{BKLM}.}%(or any other combinatorial game) have not been studied before in the literature. For a recent survey on {\sc Subtraction Games}, see \cite{LS}.

%One particular class of combinatorial games is \emph{Cumulative Subtraction games}, which 
Our general sum extension of CS generalize both Nim and the examples we consider in this section. Recall the classical ruleset family {\sc Subtraction Games}. 
%\subsection{Subtraction Games} The models we develop in this paper are motivated by the classical CGT {\em Subtraction Games}. 
The overarching idea is that there is a heap of pebbles and two players, who alternate removing pebbles given a common subtraction set $\s\subset\N=\{1,2,\ldots \}$. If the size of the heap is $x\in \N_0=\{0,1,\ldots \}$, then the current player acts by removing $s\in \s$ pebbles, and leaves the position $x-s\ge 0$ for the opponent. If there is no such $s\in \s$, then the game ends, and the result is determined, by some prescribed convention, for example, according to the above mentioned normal play convention. %, a player who cannot move loses. 
For a recent survey on {\sc Subtraction Games}, see \cite{LS}. In a more general setting, the set $\s$ may depend on the heap size, and in those cases, we insted refer to a given subtraction-map  $\s:\N_0\rightarrow 2^\N$, where, for all $x$, $\s(x)\subseteq [x]=\{1,\ldots , x\}$. More generally the subtraction set may also depend on the player, and we will return to such situations.

% Thus, from a CGT perspective, we are proposing a much delayed jump from zero-sum to general-sum games in the spirit of the historical jump in EGT, mastered first by von Neumann \cite{N} (zero-sum) and then by von Neumann and  Morgenstern \cite{NM} (general-sum), and others. From an EGT point of view, we propose a method to analyze economic style generalized combinatorial games using concepts from CGT, such as  the {\em outcome function} and {\em disjunctive sum} play \cite{Si}. In CGT, games can be naturally added and compared \cite{Si} and we show how these ideas carry over into our more general setting. 

%Such games are well studied in an elegant niche of classical Game Theory, called Combinatorial Game Theory (CGT).  %We justify our definitions with examples, and we provide theorems with short and straight-forward proofs. 
In order to illustrate the basic concepts and motivate the later theory building, we begin by investigating six  concrete situations. We consider a general ruleset family, {\sc Pebbles}, where two players, Alice and Bob, alternate in removing identical objects, say pebbles, from a common heap, given some restriction on the available actions, and given some `winning condition'. Game play ends if the current player is not able to play, because every move would lead to a negative heap size.  %Usually, the players compete in achieving a certain goal, such as removing the last pebble, or grabbing the largest number of pebbles, and the final result is usually sensitive to who starts. 
In the first five examples, at their turn, a player, Alice or Bob, must take either two or three pebbles from a single heap, so these are essentially variations of a CS game with $\s=\{2,3\}$. Later we will allow for $n\ge 2$ players. %In the sixth example, the possible move options depend on the history of play, via updated  {\em cumulations}.\footnote{Game Theory at large has a means to deal with any notion of ``history dependency'' by simply including the relevant part of the history into the notion of a position. And that is the approach we follow here as well. Of course, complexity issues increase when play depends on previous actions in any way, and in CGT one usually aims at as simple rulesets as possible, while all complexity should be in the quest of finding a ``winning strategy''. Indeed, this is the appeal of the game of {\sc Nim}. As we will see, here all information about how the players arrived at a certain game position will be erased, apart from a certain `cumulation tuple', which we feel is quite a modest complication while the gain can be huge.} 
The ultimate goal of the game will vary.  

\paragraph{P1) {\sc Normal Pebbles}.} A player who at their turn cannot move loses. Thus, if Alice starts from a heap of size 4, she should remove 3 pebbles.

\paragraph{P2) {\sc Mis\`ere Pebbles}.} A player who at their turn cannot move wins. Thus, if Alice starts from a heap of size 4, she should remove 2 pebbles. 

\paragraph{P3) {\sc Scoring Pebbles}.} A common score is updated during play by the number of pebbles the players remove; Alice's removals add to the current score while Bob's removals subtract from the current score. %The player with a larger number of pebbles when the game ends, wins. Or more precisely, 
 That is, Alice is the maximizer, whereas Bob is the minimizer;  one may think of a positive final score as a win for Alice, a negative final score as a win for Bob, and a zero final score as a tie. If Alice (Bob) collects 2 pebbles then the score increases (decreases) by 2, etc. Hence, if Alice starts from a heap of size 4, she should remove 3 pebbles, the game ends, and the final score is 3. If they remove 2 pebbles each, the game ends in a tie with a total score of $2-2=0$. %This general class of games is dubbed Cumulative Subtraction in a recent paper \cite{CLMW}. 

\paragraph{P4) {\sc Squirrel Pebbles}.} This is self-interest cumulative play. Here, there is no winner, but each squirrel attempts to gather, i.e. {\em cumulate}, the largest possible number of pebbles (nuts) for themselves. Thus, in our example, the first squirrel (Alice) should remove 3 pebbles, and the final utilities will be $(3,0)$. This is not a zero-sum game; all partial cumulations, and in particular the final cumulations, are ordered pairs of nonnegative integers. If Alice starts instead from a heap of size 7, she should collect 2 instead of 3 pebbles. Why? (See Figure~\ref{fig:ext_23-squirrel}.) And this holds for normal play and scoring play too, but in mis\`ere play you will lose regardless of how you play from a heap of size 7.

\paragraph{P5) {\sc Auction Pebbles}.} 
An auctioneer has set up the following two-player auction: the heap consists of 4 bidding-pebbles and a pair of initial bid cumulations, and the players may increase their bids by collecting either 2 or 3 bidding-pebbles. Each player has a utility function of the form: 0 utility if they do not win the auction, and otherwise  the utility is $4 $ minus `their cumulated bid', i.e. in case Alice wins the auction, $4$ minus `her initial bid plus all her play bids'. (In case of a tie, both players get utility 0.) If the initial bids are $(0,0)$ then Alice, playing first, should bid 3 to win the auction, and the utilities will be $(1,0)$. However, if the initial bid is $(1,0)$, then Alice should bid 2, because a bid of 3 would (in spite of winning the auction) give utility $4-3-1=0$, whereas a bid of 2 suffices to win the auction, and her utility will be $4-2-1=1$. Therefore the best-play bid, from a heap of size 4, may depend on the initial bids, and, as we will discuss further, such situations cannot happen in the 4 first examples. This game will be revisited in Example~\ref{ex:auc}. 

\paragraph{P6) {\sc Wealth Pebbles}.}
This is again normal play; the game is played as {\sc Nim} but the players cannot remove any number of pebbles that exceeds their current pocket cumulation. Any removal adds to the current player's pocket. Suppose that the heap size is 3, and the current cumulation is $(2,2)$. Then the first player loses. If the heap size is 3, and the current cumulation is $(2,1)$, then Alice will win if she starts by removing 1 pebble, but she loses if she starts by removing 2 pebbles. Suppose next that the heap is of size 6, and the initial cumulation is $(1,1)$. If Alice starts, then the next position is $5, (2,1)$, followed by Bob, playing to $4,(2,2)$. Now, Alice loses if she removes 2, but she wins if she removes 1.\\% (Indeed, this is a zero-sum game).\\

%https://arxiv.org/abs/2501.07239
%\begin{remark}
%A ruleset family {\sc Wealth Nim}, which generalizes {\sc Wealth Pebbles}, is studied in \cite{Ankita}. In that work, the instance {\sc Robin Hood}, in which the opponent pays for the current player's move, is solved using temperature theory, which quantifies the move advantage of the starting player \cite{Si}.
%\end{remark}
\begin{remark}
A ruleset family {\sc Wealth Nim}, generalizing {\sc Wealth Pebbles}, is studied in \cite{Ankita}. That work considers the instance {\sc Robin Hood}, in which the opponent pays for the current player's move. These particular rules increase the heat of {\sc Wealth Nim}. A complete solution is given in terms of {\em temperature theory}, which quantifies the move advantage of the starting player \cite{Si}.
\end{remark}

All six example games P1 to P6 can be easily encoded as Extensive Form Games, once a starting player is announced: each game state contains a current heap size and player cumulations, and has at most two descendants;  see Figure~\ref{fig:ext_23-squirrel} for P4. Both play and the results of P1 and P2 are independent of cumulations (and thus can be set to $(0,0)$ for example).

The first three winning conditions P1-P3 have received much attention in the CGT literature. In all of them both the allowed moves and the optimal move of the current player are dictated by the size of the heap alone. In this case, the full extensive form used in EGT seems excessively large and inefficient; the game in Figure~\ref{fig:ext_23-squirrel} requires only 11 states instead of 7, but in general we would need $\Theta(x^2)$ states instead of $n$ possible heap sizes. (The relevant heap sizes are $\{0,1,\ldots,x-2,x\}$.) %\rmr{$n$ is already the number of players. Maybe switch to $k$?} 
In Example~\ref{ex:cumsub}, we develop a much more efficient {\em tabular approach} to compute the utility of arbitrary heap sizes of such self-interest two-player games; this approach generalizes a standard CGT method, which is briefly reviewed in Section~\ref{sec:motivnorm}. % (see also a recent survey on {\sc Subtraction Games} \cite{LS}). 
In fact, this type of complexity issue is at the core of this study, in a sense, the standard EGT game three approach is not a convenient tool to study EGT generalizations of {\sc Subtraction Games}. However, the further we depart from the classical CGT setting, the more we must instead rely on computations as exemplified in this figure.

%}\uln{Observe that this is true for P1 to P4. For P5 the situation is more complicated, because initial cumulations could be any pair of numbers, and since, ass illustrated,  the PSPE computation requires information about the cumulations. The extensive form game tree would be the same as in Figure 2, but the PSPE is sensitive to the initial bids. For P6, obviously the extensive form game tree depends on the initial pocket contents. If we fix a heap size, but not the pocket sizes, in order to claim an understanding of the game, one has to list all relevant extensive form game trees.}

We believe the {\sc Squirrel}, {\sc Auction} and {\sc Wealth Pebbles} situations are new to combinatorial games' study. 
Indeed, in {\sc Squirrel Pebbles} (P4) a {\em PSPE-move} still only depends on the size of the heap, but it is no longer a zero-sum game. 

In {\sc Auction Pebbles} (P5), the situation becomes more involved. As the example above shows, even for a fixed heap size, the PSPE-move may depend on the amount of pebbles collected by each player so far, so this information is no longer redundant.

{\sc Wealth Pebbles} (P6) adds even more complication, as even the \emph{move options} depend on the collected pebbles.

%Observe that, in {\sc Auction Pebbles}, although no part of the ruleset depends of the current cumulation (initial bid and so on), play in Pure Subgame Perfect Equilibrium (PSPE, every player selects the action maximizing her utility, in every subtree of the game) depends on the current cumulation. 
%Before continuing, the reader may wish to justify that the first four examples do not exhibit such behavior; in fact, as we will see, the first four examples (generalized) have smaller complexity due to the fact that their utility functions have simpler formulations. We emphasize that, in the {\sc Auction Pebbles} example, we chose the most direct utility function to justify the goal of that game. 

%\tikzstyle{vertex}=[circle,fill=black!20, minimum size=40pt, inner sep=0pt]
\tikzstyle{edge} = [draw,thick,-]
\tikzstyle{weight} = [font=\small]
\begin{figure}[!ht]
\centering{
\begin{tikzpicture}[scale=1.1, auto,swap]
\tikzstyle{vertex}=[circle,fill=gray!10, minimum size=50pt, inner sep=0pt]
    % Draw the vertices
    %Outcome 0
    \begin{small}
        \foreach \pos/\name in { {(0,7)/S}}
        \node[vertex] (\name) at \pos {$7;(0,0)$, A};    
        \foreach \pos/\name in {{(-2,5)/B}}
        \node[vertex] (\name) at \pos {$5;(2,0)$, B};
        \foreach \pos/\name in {{(-4,3)/BB}}
        \node[vertex] (\name) at \pos {$3;(2,2)$, A};
        \foreach \pos/\name in { {(-6,1)/BBB}}
        \node[vertex] (\name) at \pos {$1;(4,2)$, B};
        \foreach \pos/\name in { {(-4,1)/BBC}}
        \node[vertex] (\name) at \pos {$0;(5,2)$, B};
         
         \foreach \pos/\name in {{(-1,3)/BC}}
        \node[vertex] (\name) at \pos {$2;(2,3)$, A};
        %Outcome 2
        \foreach \pos/\name in {  {(-2,1)/BCB}}
        \node[vertex] (\name) at \pos {$0;(4,3)$, B};
        \foreach \pos/\name in {  {(2,5)/C}}
        \node[vertex] (\name) at \pos {$4;(3,0)$, B};
 `````%Outcome 3
        \foreach \pos/\name in { {(1,3)/CB}}
        \node[vertex] (\name) at \pos {$2;(3,2)$, A};
        \foreach \pos/\name in { {(4,3)/CC}}
        \node[vertex] (\name) at \pos {$1;(3,3)$, A};
    \foreach \pos/\name in { {(0,1)/CBB}}
        \node[vertex] (\name) at \pos {$0;(5,2)$, B};
           
    % Connect vertices with edges and draw rewards
    \foreach \source/ \dest /\weight in {S/B/2, S/C/3,
                                         B/BB/2, B/BC/3,
                                         BB/BBB/2,BB/BBC/3,
                                         BC/BCB/2,C/CB/2, C/CC/3, CB/CBB/2
                                        }
\path[edge,thin] (\source) -- node[weight] {$\weight$} (\dest);
  %  \foreach \source/ \dest /\weight in {S/B/2, 
   %                                      B/BC/3,
    %                                     BB/BBC/3,
     %                                    BC/BCB/2,C/CC/3, CB/CBB/2
      %                                  };
%\path[edge,double] (\source) -- node[weight] {$\weight$} (\dest);
%\node at (-8.5,7) {$\tilde p=1$};
%\node at (-8.5,5) {$\tilde p=2$};
%\node at (-8.5,3) {$\tilde p=1$};
%\node at (-8.5,1) {$\tilde p=2$};
\end{small}
\end{tikzpicture}
}
\caption{
The picture is an Extensive Form Game tree, and it illustrates the self-interest game  {\sc Squirrel Pebbles} where the squirrels remove 2 or 3 nuts, and where the initial heap size is 7; Alice (A) starts, and the players take turns. Each node shows the heap size, the current cumulation for each player, and the current player. In Figure~\ref{fig:ext_23} we compute the PSPE of this grounded position.\label{fig:ext_23-squirrel}}
\end{figure}

{\sc Squirrel Pebbles} P4 is obviously not zero-sum, and although very simple, we believe that the general class of such self-interest games has not yet been studied in the literature of {\sc Subtraction Games} (or elsewhere).  This is probably due to the fact that the combinatorial game tradition is rooted in recreational games,  and most recreational two-player games are considered as win/loss/drawn games.\footnote{Note that some classical recreational $n$-player games have ingredients of squirrel play and/or {\sc Wealth Pebbles}. E.g. in {\sc Monopoly} players try to accumulate as much wealth as possible, and the ranges of available actions depend on how much they have accumulated so far. Of course {\sc Monopoly} is not a combinatorial game, because it has random moves and hidden cards, but the analogy is close enough to motivate a natural class of wealth games in our discussions. Note also that a player who cannot play because they are ruined are removed from the game, so it is in a sense close to the normal play convention.} 

%All six examples may be included to a super class of combinatorial games. 

We expand the framework of combinatorial games to include P4-P6, among others, while matching the existing literature of Game Theory at large, and in particular embracing the conventions of Extensive Form Games. This is a natural way forward, since all Combinatorial Games, with a given starting player, are Extensive Form Games. As a consequence, our approach provides an additional set of tools to Extensive Form Games such as a variation of starting player, outcome functions and the disjunctive sum operator.  %\rmr{We say we include all 6 examples and then say only first 5. Which is it?}\uln{We address quickly P6 in a paragraph above. The model we are building include all six variants, but then again, the most interesting part of the model is to understand what distinguishes P5 from P1 to P4. If players can ignore the cumulations in searching for `best' play, then this is both a theoretical and computational benefit. Obviously in P6, nobody can ignore the cumulations, since they are dictating the rules.}

 We develop a ruleset family large enough to encompass any Extensive Form Game, and small enough to not obscure the main direction, and distinctions we wish to address. 
  The first main result (Theorem~\ref{thm:main1}) studies a vast generalization of the CGT-outcome function, and it includes all six {\sc pebbles} variations, and much more.   
 The second main result (Theorem~\ref{thm:main2}) studies a ruleset restriction, which we will call {\sc Heap Dynamic}, that includes P1-P4 but (probably) not  {\sc Auction Pebbles} P5 and (definitely) not {\sc Wealth Pebbles} P6. The extensions to P5, where equilibrium strategies can depend on cumulations, and P6, where rules  depend on  cumulations, are essential for various modelling purposes, but they may be less efficient in terms of computing utilities and outcomes.

\subsection{The three layers of a Cumulative Game Form}\label{sec:layers}
%\paragraph{Bridging the Fields} We think of an $n$-player combinatorial game position in terms of an $n$-tuple of EG states, one for each starting player, $p^0\in \n$. Here, we already encountered an important distinction between EG and CG, a combinatorial game position is defined with respect to any starting player (Layer~2), whereas an EG state has assigned a given player to start (Layer~3). For an EG there is a priori no immediate/obvious way to attach another starting player to a game state. For a combinatorial game, it is required that any player may start from any given game position. For example, if we view a given Chess position  (without asking the question ``who is to play?'') the position may often apply to any player as starting player (apart for special situations, such as if one of the player's King is in check). A typical CGT-question for a given two-player combinatorial game position is: ``Would you prefer to start, or going second?'' There are several more or less subtle distinctions between the two fields (that share the same heading ``Game Theory''). In this study we will assume perfect and complete information, thus remaining close to standard CGT concepts. 
As mentioned, a key difference between EGT and CGT is that the EGT description of a game requires a particular initial state together with a specification of which player plays in every game state.  

In contrast, for various theoretical reasons, CGT a priori does not fix such information. If we claim to understand a CGT ruleset, we must understand every position (as a potential starting position), and if we claim to understand a specific position, we require information about the optimal moves for any starting player. At a higher level, if we claim to fully understand a particular position, we must understand how it behaves when played in a disjunctive sum with other games (within some super class of games). In CGT, whenever the terminology `game value' is used it refers to a simplest form of an equivalence class of games at this level; this is known to exist in the so-called normal play setting, but might not exist in other conventions.\footnote{Recent research; in particular Siegel's study in \cite{Si3} dwells a lot further on these type of problems.} 

On the other hand, an EGT `game value', if it exists, refers to an equilibrium in a particular game instance, as in the first paragraph. %There are more differences, such as: EGT Extensive Game Forms are prepared for  hidden information,  

In order to bridge this gap, we will advance slowly, by first defining an $n$-player {\em Cumulative Game Form} (CGF) by using three layers, to be specified in the various settings (Definitions~\ref{def:cumsub} and \ref{def:cumgameform}). In this way we can apply both CGT and EGT concepts. Let us sketch the idea of the three layers here, in full generality, i.e. in a setting with $d\ge 1 $ heaps and $n\ge 2$ players.%,  elobertaed just below. %As an overview, layer~1 gives the general rules of game, layer~2 is the de facto standard where CGT concepts thrive, and layer~3 ditto for EGT.
\begin{description}
\item[Layer~1.] The \emph{heap space} $\Omega(d,n)$ contains all $d$-tuples of finite heap sizes, where each heap, in every tuple, memorizes an $n$-tuple of cumulations;\footnote{In this layer, we use the term ``memorize'', because as we will see later, the cumulations will be updated via certain reward functions. For full generality, we require a cumulation vector on each heap, for example if we want to model a disjunctive sum play, where all heaps are independent. See the ``layered example'' just below.} 
\item[Layer~2.] A \emph{heap position}, $\bs\omega\in\Omega$, is an instance of Layer~1,  together with an $n$-player ruleset \eqref{eq:recursive};% which enables a recursive construction of heap tuples from the heap space;%, together with a cumulation tuple for each heap;%, and it applies to any starting player;

\item[Layer~3.] A \emph{grounded position} is a heap position, or a {\em disjunctive sum}  \eqref{eq:disj} of heap positions,\footnote{There are other sum-operators in the CGT-literature, but the disjunctive sum is by far the most popular.} together with a {\em current player}, which is specified via  a {\em turn function}.
\end{description}

%\rmr{we should also differentiate the labels between Alice and Bob's edges. Maybe salted text or overline/underline or left/right arrow accent?}\rmr{Also I think showing one node with its children is enough, and then we can take one of the nodes that already appears in the previous figure}

 %An $n$-player {\em ruleset} is specified together with the three layers. Altogether we have a CGF. 

%Let us describe the recursive construction of a heap position here. 
%In we describe in detail what we mean by that 
%It is convenient to identify the current player with the turn function (see Section~\ref{sec:CGF} for more details). 
For each player, for each heap position, the ruleset in Layer~2 describes their set of options.\footnote{A {\em ruleset} describes what each player can do at their turn (such as an instance of {\sc Pebbles}).  For each player, it specifies what actions they can take on the heaps, which may depend on the cumulations (as in P6 {\sc Wealth Pebbles} above), and what consequently happens to the individual cumulations on the heaps.} Thus, we can identify a heap position  with the unique $n$-tuple that describes the options for all players, i.e. 
\begin{align}\label{eq:recursive}
\vec\omega=\left(\vec\omega^i\right)_{i\in [n]}, 
\end{align}
where $\bs\omega^i$ denotes the set of options of Player~$i$ on heap position $\bs\omega$. Each option belongs to the same heap space. %\footnote{We study so called short games, so there are no cycles and no infinite play sequences, independently of order of play.} 
Note that this construction does not require the notion of a turn function.

Let us illustrate the three layers with a CGT-type game tree (see more examples in Section~\ref{sec:motivnorm}). The \textbf{Layer~1} heap space includes all the $\binom{7+1}{3}$ partitions of $7$ pebbles into the heap and the pockets of player~1 and 2.  A \textbf{Layer~2} heap position is one such partition, together with the (at most) four partitions by moving 2 or 3 pebbles from the heap to one of the players, see Fig.~\ref{fig:CGT_23squirrel4}. Finally, a \textbf{Layer~3} grounded position also specifies the player, and thus corresponds to single node in Fig.~\ref{fig:ext_23-squirrel}. Note that there are other grounded positions not shown since they cannot be reached from the root, e.g. $6;(1,0)$).    

%Notably, the two middle terminal cumulations cannot appear in alternating play, i.e. if we ground this heap  position (compare with Figure~\ref{fig:ext_23-squirrel}).
%\tikzstyle{vertex}=[circle,fill=black!20, minimum size=40pt, inner sep=0pt]
\tikzstyle{edge} = [draw,thick,-]
\tikzstyle{weight} = [font=\small, above]
\begin{figure}[!ht]
\centering{
\begin{tikzpicture}[scale=1.1, auto,swap]
\tikzstyle{vertex}=[circle,fill=gray!10, minimum size=40pt, inner sep=0pt]
    % Draw the vertices
    %Outcome 0
    \begin{small}
        \foreach \pos/\name in { {(0,0)/S}}
        \node[vertex] (\name) at \pos {$4;(0,0)$};    
         \foreach \pos/\name in {{(-3,-2)/B}}
         \node[vertex] (\name) at \pos {$2;(2,0)$};
         \foreach \pos/\name in {{(-1,-2)/C}}
         \node[vertex] (\name) at \pos {$1;(3,0)$};
          \foreach \pos/\name in {{(3,-2)/D}}
         \node[vertex] (\name) at \pos {$2;(0,2)$};
         \foreach \pos/\name in {{(1,-2)/E}}
         \node[vertex] (\name) at \pos {$1;(0,3)$};
         \foreach \pos/\name in { {(-4,-4)/F}}
         \node[vertex] (\name) at \pos {$1;(2,2)$};
         \foreach \pos/\name in { {(-2,-4)/G}}
         \node[vertex] (\name) at \pos {$0;(4,0)$};
         \foreach \pos/\name in { {(4,-4)/H}}
         \node[vertex] (\name) at \pos {$1;(2,2)$};
         \foreach \pos/\name in { {(2,-4)/I}}
         \node[vertex] (\name) at \pos {$0;(0,4)$};
         
 %         \foreach \pos/\name in {{(-1,3)/BC}}
 %        \node[vertex] (\name) at \pos {$2;(2,3)$, A};
 %        %Outcome 2
 %        \foreach \pos/\name in {  {(-2,1)/BCB}}
 %        \node[vertex] (\name) at \pos {$0;(4,3)$, B};
 %        \foreach \pos/\name in {  {(2,5)/C}}
 %        \node[vertex] (\name) at \pos {$4;(3,0)$, B};
 % `````%Outcome 3
 %        \foreach \pos/\name in { {(1,3)/CB}}
 %        \node[vertex] (\name) at \pos {$2;(3,2)$, A};
 %        \foreach \pos/\name in { {(4,3)/CC}}
 %        \node[vertex] (\name) at \pos {$1;(3,3)$, A};
 %    \foreach \pos/\name in { {(0,1)/CBB}}
 %        \node[vertex] (\name) at \pos {$0;(5,2)$, B};
           
    % Connect vertices with edges and draw rewards
    \foreach \source/ \dest /\weight in {S/B/{\hspace{-2mm}2}, S/C/{\hspace{-2mm}3}, S/D/{\hspace{2mm}2}, S/E/{\hspace{2mm}3}, B/F/{\hspace{-3mm}2}, B/G/{\hspace{3mm}2},   D/H/{\hspace{3mm}2}, D/I/{\hspace{-3mm}2}        
                             }
\path[edge,thin] (\source) -- node[weight] {$\weight$} (\dest);
  %  \foreach \source/ \dest /\weight in {S/B/2, 
   %                                      B/BC/3,
    %                                     BB/BBC/3,
     %                                    BC/BCB/2,C/CC/3, CB/CBB/2
      %                                  };
%\path[edge,double] (\source) -- node[weight] {$\weight$} (\dest);
%\node at (-8.5,7) {$\tilde p=1$};
%\node at (-8.5,5) {$\tilde p=2$};
%\node at (-8.5,3) {$\tilde p=1$};
%\node at (-8.5,1) {$\tilde p=2$};
\end{small}
\end{tikzpicture}
}
\caption{ The picture is a CGT-type game tree, and it illustrates a heap position of the self-interest game  {\sc Squirrel Pebbles} where the squirrels remove 2 or 3 nuts, and where the initial heap size is 4; each node shows the heap size, and the current cumulation for each player. The left slanting edges represent Alice's options, whereas the right slanting edges represent Bob's options.}\label{fig:CGT_23squirrel4}
\end{figure}
%Sometimes we identify the action map and the turn function as a ruleset, while assuming an underlying heap space. Sometimes we include a utility function to the notion of a ruleset; the surrounding context decides the meaning of a ruleset.% or a ruleset family.

\paragraph{Disjunctive sum}
Now, we are prepared to give the recursive description of the disjunctive sum operator, ``$+_{\! n}$''. The subscript `$n$' accentuates that we have to fix the number of players if we want to add heap positions, while, as usual,  the number of heaps may vary. Suppose that we have two heap positions, $\bs\omega_1\in\Omega(d_1,n)$ and $\bs\omega_2\in\Omega(d_2,n)$. Then 
\begin{align}\label{eq:disj}
\bs\omega_1+_{\! n}\bs\omega_2=\left(\left(\bs\omega_1+_{\! n}\bs\omega_2^i\right)\cup \left(\bs\omega_1^i+_{\! n}\bs\omega_2\right)\right)_{i\in [n]},
\end{align}
where e.g. $\bs\omega_1+_{\! n}\bs\omega_2^i=\{\bs\omega_1+_{\! n} \vec\omega' \mid \vec\omega'\in \bs\omega_2^i\}$. The disjunctive sum $\bs\omega_1+_{\! n} \bs\omega_2$ belongs to the heap space $\Omega(d_{1}+ d_{2},n)$; see Section~\ref{sec:disum} for further discussions. % (it also helps a human eye to distinguish it from usual addition).
%\rmr{continuation of critical issue: the + is defined over a pair of particular positions. But what we really want is to join rulesets. If I get it right then the new Layer~1 is the cross product of the two smaller games, and the new ruleset is obtained by collecting all joint positions }\uln{The + is standard in CGT, and moreover: at the end of game, we will add up the utilities from the respective heaps. If you want to change this we could discuss first.}

%an $n$-player disjunctive sum of games here, since it is a motivational factor of the second layer, and since it is part of the third layer.

%Consider two heap positions $G, H$, then $G+H$ is the heap position 

%In  Section~\ref{sec:disum} we continue this discussion, but 

%That is, each grounded position comes with a set of current player move options in terms of a subset of the heap space, and if this set is empty the game has terminated. 
The {\em turn function} determines the current player (at each node in the play sequence). Usually, in CGT terms, it depends only on the previous player, but for full generality (as required by Extensive Form Games) the heap space will be included in the domain.

%\rmr{I don't understand why the disjunctive sum is part of the Layer~3 definition. It looks like all we need for Layer 3 is Layer 2 and a current player. 
%Then when defining disjunctive sum (on Layer 2) then the sum of two Layer 2 positions is just a new Layer 2 positions. Since any position (plus player) can become a grounded position, this applies to a sum of positions as well. No need to specify this in the grounded position definition.}\uln{This is at the core of CGT. A CGT grounded position can be a disjunctive sum of game components. It fits our model very well. }

The $d$ heaps may be independent, in the sense that every possible move for every player affects precisely one heap. In this case, the situation is identical to disjunctive sum play, and we think of the total ruleset as a Cartesian product of the individual heap rulesets. %Moreover, is given.

A Cumulative Game Form is played by declaring a grounded position, which may be a disjunctive sum, and which includes a turn function as on Layer~3. 
Note that a game cannot be played on the first two layers; these are theoretical building blocks. % although sometimes, the ruleset alone, which requires only Layer~1 is called `a game'. 
The heap positions on Layer~2 are sometimes called \emph{game components}, as they are prepared for  disjunctive sum play, together with other game components. We cannot apply the disjunctive sum operator on Layer~3, but we can instead ground a disjunctive sum game in order to play it. This gives room for plenty flexibility: we may ground a single component, or we may, at some later point, ground a number of heaps with various rulesets, to merge them into a larger ruleset. 

The reader might have observed two issues that require further discussion. 
\begin{itemize}
\item[1)] Will game play, from a grounded position, terminate?  
\item[2)] Any grounded  position can be played, but yet without any incentive. 
\end{itemize}
To address the first item, we use a standard CGT property, that allows standard recursive constructions of game forms, outcomes, values etc. %One can think of a CGT {\em ruleset} in an abstract way, as countably many digraphs with multicolored edges (one color for each player). For each instance, a token is placed on a node, and given some turn function, the players slides the token, until perhaps the game terminates because the current player cannot move. At that point some `winning condition' is declared. 

\begin{property}[Ruleset Feasibility]
A ruleset is {\em feasible} if, for every starting position, a game terminates independently of order of play.
\end{property}
Thus {\em feasibility} is a property of the game's Layer~2 description. In particular it guarantees that a disjunctive sum of any two positions with feasible ruleset is also feasible. %\rmr{note edited text. see my critical comments above} %does not depend on the turn function. This is important, because even if we fix all three layers, we do not know the order of play in specific components of a disjunctive sum of games.\footnote{In this study, we do not assign utilities for infinite components, and we do not define outcomes for cyclic/infinit play.}

Regarding the second item, when we add a utility function to a game form, it becomes a game. Various incentives will be given by utility functions, and player maximization of individual utilities is the purpose of a game. As we will see in Section~\ref{sec:extform}, a grounded position, together with player utilities, is an (EGT) Extensive Form Game. Here, we wish to emphasize the distinction between the three layers, and in particular the often overlooked gap between Layers~2 and~3. When a utility function is added to a CGF, we call it a  {\sc Cumulative Game}.\footnote{A {Cumulative Game Form} is a ruleset family, where any utility function may be added later. A {\sc Cumulative Game} is. a ruleset family where utilities have been specified. }

\paragraph{Outcome classes}

Zermelo's theorem \cite{Ze} is the foundation of both the outcome classification in combinatorial game theory and the recursively computed pure-strategy perfect equilibrium (PSPE) in extensive form game theory. Despite this shared origin, the two game-theoretic frameworks have diverged over time. Here, we bring them back into contact. In our setting, Zermelo’s result still applies, provided that, in the case of indifference between options, each player possesses a strict preference ordering over the other players' results, including whether they should act friendly or not. Our generalized CGT outcome function computes a vector of PSPE values, depending on the starting player and a specified turn function; each entry should be evaluated (using backward induction) based on  when the current player cannot move.

Let us revisit our layer construction in light of defining a CGT-style outcome function. The heap positions on Layer~2 are more sophisticated than they might initially appear. Taken independently, they are already well-suited to disjunctive sum constructions. With additional structure provided by Layer~3, we can assign them generalized CGT outcomes that, in some cases, coincide with generalized EGT (PSPE) utilities. See our main results in Section~\ref{sec:valout}.

\paragraph{Ruleset.} Here we use the term {\em ruleset} in the sense of Layer~2. At this level, any player knows what they can do, as a current player, and the consequences of the actions are revealed in terms of an immediate update of cumulations; later we use a {\em reward} function for this purpose. But ``a ruleset'' is also standard language, and we will use it in other contexts as well, for example if we want to include the information about the turn function and/or the utilities of the players. All this is usually understood when the term  ``ruleset'' is used. In a sense, the Layer~2 definition of a ruleset is minimalistic: that is the least information we have to give.\\

Sometimes the rules do not depend on whose turn it is. 
\begin{property}[Symmetric Rules]\label{prope:sym}
Rules are \emph{symmetric} if, for all positions, the move options are independent of player.\footnote{In CGT, the common term for symmetric is ``impartial''; but those games are not `impartial' in the economic sense, so we avoid this traditional CGT terminology here.} %In this case, we write $\s(x,p)=\s(x)$.
\end{property}

\subsection{A layered disjunctive sum example}\label{sec:laydissum}
Suppose we combine P2 (mis\`ere play) and P4 (squirrel play) in a two-player alternating play game. Thus the heap space contains all pairs of non-negative integers $(x_1,x_2)$, where each $x_i$ comes with a pair of cumulations  $(C_1^i,C_2^i)$. 
So a typical heap position takes the form $$((x_1,(C_1^1,C_2^1)), (x_2,(C_1^2,C_2^2)).$$ 
Since, we are modelling a situation with individual components, this position can also be written as $(x_1,(C_1^1,C_2^1)) \, +_{2}\, (x_2,(C_1^2,C_2^2))$, by using the two-player disjunctive sum operator `$+_{2}$`.

As usual in disjunctive sum theories, `best play' in individual components will not necessarily guide us in understanding `best play' in the sum of the games (for those CGT readers, the ``outcome'' of the ocmponents may not suffice to compute the ``outcome'' of the sum).

A typical grounded position is of the form $$((x_1,(C^1_1,C^2_1))+_2(x_2,(C^1_2,C^2_2),p),$$ where $p$ is the previous player. 

Now, if we would like to invite Alice and Bob to play this disjunctive sum game, we may set $p$ to Bob, so Alice starts from say the grounded position 
\begin{align}\label{eq:exground}
((4,(0,0)), (4,(0,0)),\rm{Bob}).
\end{align}
Recall that the rulesets on the individual heaps are: take either 2 or 3 pebbles from precisely one of the heaps, and add whatever you take from a heap to that heap cumulation. The mis\`ere play component is obviously cumulation independent: the utility will be negative for the player who makes the last move in that component, and positive for the other player, say $-1$ and $+1$ respectively.\footnote{We are turning normal and mis\`ere play into scoring play here in order to be able to add utilities in a set up with general sum games. From a theoretical point of view, some confusion might arise in that we must assign a mis\`ere utility to a component that is already terminal (without any last move). This can be modelled by allowing a last move to a negative heap size on that particular component, that will give a negative utility. But for the purpose of this initial example it is not important, since there will be a last move on that component, which in mis\`ere play should be punished with a negative utility.} The self-interest component will have the natural cumulative scoring principle: ``you get what you take'', 2 or 3 pebbles, until the game ends, and, for each player, the final utility will be the sum of the heap  utilities.

Suppose that Alice starts by removing three pebbles from the self-interest component in \eqref{eq:exground}. Now the position is:
$$((4,(0,0)), (1,(3,0)),\rm{Bob}),$$
%and there are two players. Then  
and no more move is possible on the second component. 
If Bob removes two pebbles from the first heap, then Alice will have to make the last move in the game, by removing two pebbles from the same heap. The final utilities will be $(2,1)$. Alice has a better opening move, via a sacrifice  (Remark~\ref{rem:greedy}) on the self-interest heap, in which case the utility pair will instead be $(3,1)$, if Bob makes the sensible move and responds on the same heap. If we instead set the respective mis\`ere utilities to $-1/2$ and $1/2$, then Alice is indifferent between greedy and non-greedy play on the self-interest heap. And if we set the utilities to $-1/4$ and $1/4$, then she should play greedily. So, best play in a disjunctive sum of games can be sensitive to the magnitude of the utilities we assign on the individual heaps.

\begin{remark}[Greedy Play]\label{rem:greedy}
%\paragraph{Greedy play.} 
\emph{Greedy} play means to play the largest action possible, and \emph{sacrifice} means to  play some smaller action. E.g. in the PSPE play sequence above there is one sacrifice, in the first move from a heap of size 7. The necessity of sacrifices in certain situations was observed in \cite{St} via a certain periodicity conjecture, which was solved in \cite{CLMW}. 
\end{remark}

To our best knowledge, these type of problems have not yet been studied in the literature. However in this study we will not dwell further on how to select utilities to model various situations, but we will instead set a framework to open up for future modelling problems.

 \subsection{A motivation from normal play CGT}\label{sec:motivnorm}
 %Before moving forward, we want to briefly explain the importance of two key concepts from CGT that this paper aims to extend: the \emph{outcome function} and the \emph{disjunctive sum} operator. Both concepts are reviewed, through the (original) setting of Normal play, at last, in Section~\ref{sec:NP}.
 The classical two-player, alternating play, normal play theory lays the ground for CGT. Let us review some aspects that motivate this study. 
 \paragraph{The normal play outcome function.}
 An outcome function assigns to each game one out of four outcome classes, which describes who wins, depending on who starts, under optimal play. Therefore it is an extension of the pure subgame perfect equilibrium in EGT, that, in CGT, opens up a variety of useful information about the game, as we explain next. 

The CGT players are usually called Left 
(positive) and Right (negative). The partizan play outcomes are $\Le$ (Left wins independntly of who starts), $\Ne$ (the curreNt player wins), $\Pre$ (the Previous player wins), and $\Rig$ (Right wins independently of who starts). These outcomes are partially ordered, with $\Le>\Ne>\Rig,\Le>\Pre>\Rig$, but where \Ne\ and \Pre\ are incomparable. In impartial play, we have only the incomparable outcomes \Ne\  and \Pre.

 \paragraph{The tabular approach.}

%\begin{example}[Outcomes of Normal play %Subtraction]\label{ex:subtr23}
Consider $\s=\{2,3\}$, as in {\sc Normal Pebbles}, P1. The initial outcomes are: 
\begin{center}
\begin{tabular}{|c| c c c c c c c c |}
\hline
$x$		&0 &1 &2 &3 &4 &5 &6 &7  \\ 
\hline
 $o(x)$ 	&\Pre &\Pre &\Ne &\Ne &\Ne &\Pre  &\Pre &\Ne  \\ 
\hline
%$\Lu$ & $o(x,\Ri)$   &$R$ &$R$ &$L$ &$L$ %&$L$ &$R$ &$R$ &$L$  \\ 
%$\Ri$ & $o(x,\Lu)$ 	&$L$ &$L$ &$R$ &$R$ &$R$ &$L$ &$L$ &$R$  \\ 
%\hline
\end{tabular}
\end{center}
They are computed recursively by checking if there is a \Pre-position as an option. If so, write an \Ne, and otherwise, write a \Pre. Observe how much simpler, and broader, this approach is than the standard method used for Extensive Form Games in Figure~\ref{fig:ext_23-squirrel}. For example, this approach does not fix the starting position, but it is reaadily applicable to all Layer~1 positions. By using the EGT approach, we would have to draw a new game tree for every new starting position. We will see in Section~\ref{sec:cumsub} how this normal play table approach can be generalized to the setting of self-interest utilities, and further down, as the main results Theorem~\ref{thm:main1} and Theorem~\ref{thm:main2}, we study how the general class {\sc {\sc Cumulative Game}s} can be restricted so that this tabular approach is still viable. %Figure~\ref{fig:ext_23-squirrel} also fixes the starting player. 
Our next table, shows that our table approach generalizes to arbitrary move options for the players, and we compute simultaneously the optimal outcome for either player as a starting player.
%\begin{example}[Outcomes of a partizan game]\label{ex:partizan}

We let $\s=(\{2,3\}, \{1,4\})$ denote a partizan ruleset where \Lu\ subtracts 2 or 3 and \Ri\ subtracts 1 or 4 if the resulting position is nonnegative. The initial outcomes are as in the second row:

\begin{center}
\begin{tabular}{|c| c c c c c c c c c c|}
\hline
$x$		&0 &1 &2 &3 &4 &5 &6 &7 &8&9 \\ 
\hline
$o(x)$ 	&\Pre &\Rig &\Ne &\Le &\Rig &\Ne  &\Le &\Pre &\Ne&\Le \\ 
\hline
$o_L(x)$ &R &R &L &L &R &L &L &R&L& L \\ 
$o_R(x)$ 	&L &R &R &L &R &R &L &L &R& L\\ 
\hline
\end{tabular}
\end{center}
These outcomes are computed recursively: write an \Ne, if Left can move to \Le\ or \Pre\ and Right can move to \Rig\ or \Pre; write a \Pre, if Left cannot move to \Le\ or \Pre\ and Right cannot move to \Rig\ or \Pre. Similarly, if only Left (Right) has a winning move as a starting player, then write a \Le\ (\Rig). For rows three and four, note that the standard CGT outcomes can be decomposed as $o(x)=(o_L(x),o_R(x))$, where $o_L(x)$ denotes the result in optimal play when Left starts, and $o_R(x)$  deontes the result in optimal play when Right starts; in the table `L' is `Left wins' and `R' is `Right wins'. By using this notation, the tabular approach will resemble better a more general approach for {\sc Cumulative Games} (see Section~\ref{sec:cumsub}): here, write $o_L(x)=\mathrm{L}$ if and only if $o_R(x^L)=\mathrm{L}$ for some Left option $x^L$, and write $o_R(x)=\mathrm{R}$ if and only if $o_L(x^R)=\mathrm{R}$ for some Right option $x^R$. %write an \Le, if there is a \Pre\ or an \Le\ as a Left  option; write an \Rig\, if there is a \Pre\ or a \Rig\ as a Right option; write \Pre, if all options are \Ne; and otherwise write an \Ne.
%Consider optimal play from a heap of size 7. The outcome is \Pre, because \Lu\ has options to $L$\ ($7-2$) and $R$\ ($7-3$), whereas \Ri\ has both options to $L$\ ($7-1$ and $7-4$). However if we start from position $x=6$, then \Lu\ wins regardless of who starts.

 \paragraph{Recursive game and disjunctive sum definitions.}
 We have defined these concepts in Section~\ref{sec:layers}. They have the same interpretation in two-player normal play games. The game forms are recursively defined via the notation $G=\{G^\mathcal L\mid G^\mathcal R\}$, where e.g. $G^\mathcal L$ is the set of Left options, and similarly for the disjunctive sum.
 
\paragraph{The CGT game tree approach.}
Below, we draw the standard CGT game tree for the starting position with three pebbles in the previous table. The CGT tradition lists all options of every sub-position of a game, and left slanting edges are used for Left options, while right slanting edges are used for Right options.  As the table shows, this is an \Le-position, but as we mentioned the outcome class does not suffice if we want to play this game in disjunctive sum with other games. The game tree approach suffices, and as we will see, we can simplify the game tree to a canonical form, in fact (for those CGT readers) $G=\{0,-1\,|\,\{0\,|\,-1\}\}={\rm Tiny}(1)>0$, where the second Left option is {\em dominated}. We emphasize that, as in the case of the outcomes, the game theoretical value is computed by not fixing a starting player (at any sub-positions). The second equality sign ``='' is explained in the next paragraph.\\ 

\begin{center}
\begin{tikzpicture} [scale = 0.6]
%\draw (0,4) -- (2,2);
%\draw (0,0.5) -- (1.5,2);
%\draw (0,0.5) -- (-1.5,2);
\draw (0,0) -- (-2,-2);
\draw (0,0) -- (-4,-2);
\draw (-2,-2) -- (-1,-4);
%\draw (-2,2) -- (-1,0);
\draw (0,0) -- (2,-2);
\draw (2,-2) -- (3,-4);
\draw (2,-2) -- (1,-4);
\draw (3,-4) -- (4,-6);
\put(-2,5) {$G_1$};
\foreach \Point in {(0,0),(-2,-2),(2,-2),(-4,-2),(3,-4),(4,-6), (-1,-4),(1,-4)}{
\fill \Point {circle[radius=2.5pt]};
}
\end{tikzpicture}
\end{center}

 %An interesting feature of combinatorial games is that they can be added. Intuitively, playing the game $G+H$ means that the current player must make a move in exactly one of those games (consider for example a situation where $G$ is some Nim position, and $H$ is a Chess opening position). Under the Normal play convention, a player that cannot move in any of these games loses the game. 
% One important question is which classes of games are closed under the disjunctive sum operation. 

\paragraph{Partial order, group structure, play comparison, reduction theorems and canonical forms.} The two CGT key concepts, ``outcome function'' and ``disjunctive sum'' induce a partial order relation over zero-sum games within a given class, where $G\geq H$ if the outcome of $G+X$ is never worse for the maximizing player, Left, than the outcome of $H+X$, independently of how $X$ is chosen. For normal play this definition turns out to be particularly powerful: every two games can be compared constructively, by essentially playing them out together. Normal play games constitute a group structure, where the negative of a game $G$, denoted $-G$, is the game where the two players have swapped positions. The main theorem of normal play says that $G\geq H$ if and only if player \Lu\ wins the composite game $G+(-H)$ playing second, and $G=H$ if and only if the previous player wins $G-H$. Moreover this result implies a bijection of the natural partial order relations with the four outcome classes, a unique feature within combinatorial games.\footnote{Although, this simple and elegant idea fails for other universes of games, it has still been demonstrated to lay a foundation for larger classes of CGT, namely, first Guaranteed Scoring games \cite{LNS,LNNS}, and then for Absolute CGT \cite{LNSabs1,LNSabs2}, that includes various mis\`ere play settings. By generalizing to general sum games (such as self-interest), however, the normal play foundation seems to break, mainly because games cease to be purely competitive, and the concept of a normal play order embedding, as paved the way forward towards Absolute CGT, does not appear to make much sense in the larger class we study here. The big open problem is what to do instead to have interesting subclasses of games with constructive game comparison, if possible at all. (Another possible direction would be to weaken the overlaying idea that games comparison must include the ``for all $X$'' part, if this can be justified by interesting applications.)}

This is a very satisfactory result in itself, and it becomes even more interesting while the reductions, above mentioned {\em domination} and {\em reversibility} are powerful reduction tools.\footnote{Please see any of the standard texts about combinatorial games for these definitions.} The game form that results after all possible reductions, in any order, is a unique simplest form, so we can talk about normal play {\em game values}, a.k.a. {\em canonical forms}. 

The above game example $G_1$, reduces by dominating one of its Left options. Let us give an example of a game that reduces the rank of the game tree (and not the width). The literal form of $G_2$ is $\{*|*\}$. However reversibility ensures that $G_2=0$. This can be verified directly by the  normal play main theorem, namely the current player loses.\\

\begin{center}
\begin{tikzpicture} [scale = 0.6]
\draw (0,0) -- (-2,-2);
%\draw (0,0) -- (-4,-2);
\draw (-2,-2) -- (-1,-4);
\draw (-2,-2) -- (-3,-4);
\draw (0,0) -- (2,-2);
\draw (2,-2) -- (3,-4);
\draw (2,-2) -- (1,-4);
%\draw (3,-4) -- (4,-6);
\put(-2,5) {$G_2$};
\foreach \Point in {(0,0),(-2,-2),(2,-2), (-1,-4),(1,-4), (3,-4),(-3,-4)}{
\fill \Point {circle[radius=2.5pt]};
}
\end{tikzpicture}
\end{center}

On a note, some subclasses of normal play games have even stronger properties. For example, every impartial (symmetric) game $G$, however complicated, is equivalent (under the above relation) to some single-pile {\sc Nim} game of size $n_G$. This means not only that the same player wins in $G$ and in $n_G$, but that we may replace the game $G$ in every possible context (e.g. $G \, +$ {\sc Domineering}) with the {\sc Nim} pile $n_G$, and the results would stay the same. The origin of this result is the classical Sprague and Grundy theory \cite{Sp, G}, where all games are impartial, and it is not hard to see that it can be extended to all normal play games. 

\paragraph{Approximation approaches.}
The equivalence classes are large in normal play CGT (for example the 0-class contains all games $G-G$), which often implies powerful reductions. Even so, the canonical forms tend to grow rapidly into, for a human eye, an incomprehensible  jumble. Luckily there are efficient approximation tools in normal play CGT, such as reduced canonical form, mean value, temperature, and atomic weights. This is another reason why normal play has become very popular. Mis\`ere play receives more and more attention, and has an extremely rich theory from other points of view. But note that the equivalence class of zero in mis\`ere play without any restrictions is the trivial one; the current literature studies a multitude of restrictions to obtain larger equivalence classes, constructive comparison, and more powerful reductions. Scoring play combinatorial games are similar to mis\`ere play in this respect. From this point of view, we are looking forward to the development of self-interest (general sum) combinatorial games in the form of {\sc {\sc Cumulative Game}s} or otherwise.

\paragraph{The current or the previous player.} The player about to move is the \emph{current player}.\footnote{The current player is traditionally called the ``next player'', which is a bit unfortunate, because in a current position, the previous player and the next player would be the same player, which is not what is intended. Hence, we will not adapt to that convention here. Whenever we use the term ``next player'' it referres to the player succeding the current player.}  The other player is the {\em previous player}. In normal play, the previous player win positions are exactly the zeros modulo equivalence, while those positions for which the current player wins are much more diverse (in fact they are all confused with zero). Alternating play is too restrictive in a typical EGT setting. The turn function will be generalized further, by adapting conventions from EGT. %; in either case the analysis will depend on the notion of a previous player. 
Extensive Form Games carries the information of a current player in every node, so the notion of a ``previous player'' have yet small relevance in EGT. %, and moreover, a ``turn function'' is usually not required. 
The hybrid developed here opens up for the possibility of mixing the best out of both worlds. 
 
 \subsection{Why cumulative play?}
 One intuitive idea of cumulative play is that ``you get what you take'', and you put it in your own pocket for later use. This idea is simpler than a typical zero-sum setting \cite{CLMW}, where the players accumulated scores must be compared to access the utility of play. In a purely self-interest setting no comparison is required. In a self-interest setting, it is more natural to allow for several players, whereas in a zero-sum setting two players is the generic model. 
 
 Accumulation of points can have various other interpretations, as is discussed in {\sc Auction Pebbles} P5 (see also Example~\ref{ex:auc}) where `cumulations' are interpreted as bids in an English Auction type setting, and where we discover that rational play might depend on cumulations, even when the ruleset does not depend on them. In a more `economic setting' one would most likely want to study {\sc  Cumulative Games} with actions depending on current accumulations, as in {\sc Wealth Pebbles}, P6. While we will not dive much deeper into the particulars of that territory in this paper (see \cite{Ankita} for results on a particular ruleset {\sc Robin Hood} in this family), these type of rulesets  are embraced by  our first main result, Theorem~\ref{thm:main1}. 
 
 By laying a framework in a vast generalization of both CGT and EGT, a natural question is ``How much must we narrow down the vision to reach some tractable class of games/rulesets?''. And, from the other point of view, one could ask ``How far we can extend {\sc Subtraction Games} and {\sc Cumulative Subtraction}?'', while still maintaining a tractable (dynamic) outcome function. We suggest that for such a class relevant information during play should only depend on the heap sizes (and not the cumulations). 
 
 The prospects of comparing (sums of) {\sc Cumulative Games} in a broader economic setting is an ultimate goal of bridging CGT and EGT. This work lays a foundation for such studies, that could be extended to include other EGT topics such as random turn games, hidden information and more. We note that, although our current model does not include `nature nodes' with a random selection of move options, our {\em reward function} to come (for cumulation updates) may be stochastic, and by regarding expected values instead of deterministic rewards, all theory presented here will still go through. 

\subsection{More on outcomes}
Outcome functions were originally defined for two-player normal play CGT; it is essentially a pair of (von-Neumann~\cite{N}) minimax algorithms, and first used in the setting of scoring combinatorial games by Milnor \cite{M}. It was reconsidered in Stewart's PhD thesis \cite{St} for symmetric {\sc Cumulative Subtraction}, and simplified in \cite{CLMW}. See also \cite{E} for a generalization of Milnor type games, and \cite{LNS, LNNS} for interesting connections of scoring combinatorial games with normal play games. %Another reference is \cite{J} for scoring games where all play sequences have the same parity. 
 Most of this work was done in the context of game comparison via disjunctive sum theory, which in every setting relies on a well defined outcome function. %See Section~\ref{sec:disum} in our setting.%where we also prove that our ruleset satisfies important properties of combinatorial games such as additive closure. 

\subsection{Contribution}
{\sc Cumulative Games} is an advanced $n$-player ruleset family, which is a framework  that combines central concepts from both CGT and EGT. 
%\rmr{Move all this to contribution subsection}\uln{ %
We build a vast generalization of the standard CGT-outcome function, now valid for any strategy profile; if the strategy profile is one in PSPE and if utility functions are self-interest, then the outcome function together with the initial cumulations correspond to an $n$-tuple of EGT game values; this is Theorem~\ref{thm:main1},  our first main result. 
Our ruleset family allows for an efficient equilibrium computation under certain restrictions. The result, Theorem~\ref{thm:main2}, which concerns the subfamily {\sc Heap Dynamic}, is the second main result of this paper; in a sense, we are presenting a vast generalization of the dynamic programming ``tabulaar approach'' that usually appears in the outcome (or nim value) computation for {\sc Subtraction Games} (Section~\ref{sec:motivnorm}) and its scoring play relative {\sc Cumulative Subtraction} (Section~\ref{sec:cumsub}).\footnote{Such tabular approaches are variations of the ``shift-register'' model, as shown in early research by Golomb \cite{Go}, with both practical and theoretical consequences. For example the model implies periodicity of one heap normal play {\sc Subtraction Games}.}  Moreover, the defined {\em disjunctive sum} operator in combination with partially ordered perfect play {\em outcomes}, let us define a partial order over games, according to the advantage that a certain player has. 
Theorems~\ref{thm:exthea} and \ref{thm:exthea2} demonstrate that {\sc Cumulative Games} encompass all Extensive Form Games.

\subsection{Outline}
In Section~\ref{sec:cumsub}, we define a class of two-player general sum combinatorial games called {\sc Cumulative Subtraction}, by using the mentioned three layered structure. We analyze this class using both EGT and CGT tools, and in particular, we introduce a CGT-inspired self-interest \emph{outcome function}, and demonstrate its usefulness. Thereby, several ideas in this paper are already present in this section in a greatly simplified form. 

The next sections are dedicated to generalizing these ideas. In Section~\ref{sec:cumgame} we present our general class of {\sc Cumulative  Games}, followed by some more background on Extensive Form Games in Section~\ref{sec:extform}. We demonstrate how {\sc {\sc Cumulative Game}s} can capture both classical zero-sum {\sc Subtraction Games} and our newer general-sum {\sc Cumulative Subtraction}, as well as many other variations. %\rmr{we said above that Cumulative subtraction is known in the literature but her we say it is ours}\uln{The zero-sum variation has been studied. Self-interest etc is ours.}

Section~\ref{sec:valout} generalizes the notion of a CGT-type outcome function, by adapting to extensive form terminology, whenever applicable. Here, we study  the conditions under which it is well defined and efficient. 

In Section~\ref{sec:disum} we return to one of the most important concepts from CGT, the disjunctive sum of games. %and the (via outcomes) induced partial order of games. 
We show that the defined disjunctive sum operator together with the defined outcome function in PSPE induce a CGT-like partial order over games. %This  could be a cornerstone in extending Sprague's \& Grundy's, Milnor's and Conway's et al. classical discoveries from zero-sum games to more general classes.

We close the loop back to EGT in Section~\ref{sec:equiv}, by showing that every Extensive Form Game is strategically equivalent to some grounded {\sc Cumulative Game}. In Section~\ref{sec:discussion}, we conclude by discussing some key issues and by giving some future directions. %In Section~\ref{sec:NP} we include a normal-play guide, in the context of {\sc Partizan Subtraction Games}, for those readers  new to the subject.

\section{Two-player Cumulative Subtraction}\label{sec:cumsub}

In this section we take an intermediate step towards full generality (Section~\ref{sec:cumgame}) and define a general-sum variation of the classical one-heap two-player {\sc Subtraction Games}. We name the new class of games by {\sc Cumulative Subtraction} (CS), although this term is also used in the special case of zero-sum games in \cite{CLMW} (which is also discussed here). %We start with two players, and generalize later in Section~\ref{sec:cumgame}. 
For two-player {\sc Cumulative Subtraction}, we call the players by Player~1 and Player~2. %\footnote{In Section~\ref{sec:NP}, the CGT tradition is emphasized by calling Player~1 by \Lu\ and Player~2 by \Ri.} 
Here, Layer~1 will be defined on one heap, while Layer~3 allows for several heaps in a disjunctive sum. The ruleset is defined on Layer~2, and it is independent of the turn function, which appears on Layer~3 as alternating play. 
 %We build onto the three layers described in the introduction. 

\begin{definition}[Two-player Cumulative Subtraction]\label{def:cumsub}
A one heap two-player Cumulative Subtraction Form is defined on three layers. % The layer are as follows. %does not seem possible: either we map the heap size only, or we map grounded positions (as you do in Eq(1)), since if we update the cumulations, we must know who the current (previous?) player is. I'm not sure how critical this is, but a possible solution is to have the cumulation vector ordered by turns, not player indices. That is, the heap position $(x,(y,z))$ would mean `current player has $y$ and previous player has $z$. Then the rule set for $S\subseteq \mathbb N$ would induce the options $\{x-a, (z,y+a)|a\in S\}$. We can extend it to more players if they rotate turns in some fixed order. }\uln{The ruleset specifies the rules for each player.  I reformulated in Section 1, to specify that the ruleset specifies, for each player, what they can do given a heap position. So we require an n-tuple of grounded positions, to define the rules of game. This is consistent with Definition 6.}
\begin{description}
\item[Layer~1.] The \emph{heap space} is  $\Omega (1,2)= \N_0\times \R^2$; 
\item[Layer~2.] A \emph{heap position} is of the form $\bs \omega=(x,(C_1,C_2))$, where $x\in\N_0$ is the heap size, and where $C_i\in \R$ is the accumulated size of Player~$i$'s pocket. The ruleset is described just below.%The ruleset $R$ is described in \eqref{eq:mo_c}; 
\item[Layer~3.] A \emph{grounded position} is a disjunctive sum of heap positions together with a current player $p$ in the alternating move convention. If played on one heap, we denote it by $(x,(C_1,C_2),p)$. 
\end{description}
The ruleset, which sometimes is referred to as {\em heap-size dynamic} \cite{HR},\footnote{Heap-size or ``pilesize'' dynamic rulesets are usually seen as generalizations of fixed or invariant rulesets; here the notion {\sc Heap Dynamic} will later apply as a restriction of the full class of {\sc Cumulative Games}.} is described via the option sets $\bs\omega_1$ and $\bs\omega_2$ from Equation~\eqref{eq:recursive}. Here they are: 

$$\bs\omega_1=\{(x-a, (C_1+a,C_2))\mid a\in \s(x)_1\}$$ and $$\bs\omega_2=\{(x-a, (C_1,C_2+a))\mid a\in \s(x)_2\};$$
%The ruleset $R$ is a pair of functions, one for each Player~$i$, $R_i:\N_0\times\R^2\rightarrow 2^{\N_0\times\R^2}$. 
where $\s:\N_0\rightarrow 2^\N\times 2^\N$ is a given {\em subtraction map},, that may depend on the heap size, but not the cumulation, and where, for all $x$, $\s(x)_1\cup\s(x)_2\subset [x]$, with in particular $\s(0)_1\cup\s(0)_2=\varnothing$. 
\end{definition}
Observe that Definition~\ref{def:cumsub} assures that every game will end as any action $a$ is a positive integer, and, for all $x,a$ the subtraction $x-a\geqslant 0$. But it does not specify who wins, or what are the player utilities. This allows for several variations,  which we discuss below. We will use the superscript `T', whenever we wish to emphasize that a position, heap size, cumulation etc. is terminal.

\begin{property}[Fixed Rules]\label{prope:fixed}
If there are subsets of the natural numbers $S_1,S_2\subset \N$ such that for all $x\in \N_0$, for $i\in \{1,2\}$, $\s(x)_i=S_i\cap \{0,\ldots ,x\}$, then the ruleset $\s$ is \emph{fixed}.\footnote{In the setting of symmetric (impartial) Subtraction Games, fixed rules have also been called Vector Subtraction Games \cite{Go}, and Invariant Games \cite{DR}.} In case of fixed rulesets, we abuse notation and denote, in case of symmetry, $\s=S_i$, and  otherwise $\s=(S_1,S_2)$.
\end{property}

\subsection{From a CGT-type interpretation}
%\rmr{what do we want to say about CGT? we started to discuss this in the intro regarding the Pebbles game. 2-player CS is just a mild generalization of Pebbles.  What is known about it (or special cases of it) in the CGT literature? Can we present some of these results using the first two layers?}
As we stated at the beginning of this section, {\sc Cumulative Subtraction} is inspired by the traditional CGT-type {\sc Subtraction Games} under normal play rules. Recall that in the normal play convention, which is fundamental to CGT, move-ability lies at the heart of the matter; it is never bad to have more moves. Utility based theories usually do not evaluate moveability. In fact, we have not yet encountered any such  discussion in the EGT-literature. Our model easily allows for a mixture of scenarios (generalizing Section~\ref{sec:laydissum}).

Closer to the current EGT models is the more recent \emph{zero-sum scoring variation}  \cite{CLMW}, where each Player~$i$ aims to  maximize the utility $C_i^\T-C_{-i}^\T$,\footnote{In EGT, the index $-i$ usually means ``the other player(s)''.} the difference of cumulations when the game ends. To align the CGT model further with  EGT, we here introduce a self-interest variant (an extension of  {\sc Squirrel Pebbles} P4), where each player aims to maximize their own terminal cumulation.  Even though the players remove objects from a common heap, there is often a clear distinction; see Section~\ref{sec:discussion} for some further discussion on this topic. 

\subsection{Towards an EGT-type interpretation}\label{sec:towEGT} 

In this section, we use two-player CS as a way to demonstrate the gap we discussed in the introduction; EGT is general in the sense that we can assign any utility function we want (e.g. some monotone function of both cumulations), and (to avoid distancing us too far from the standard CGT settings) the PSPE will be well defined with a unique outcome. Thus, for example,  we can  generalize the standard CGT model, to add and compare games. Note though that this is a first step, and one could envision various ways forward, for example with hidden information or random turns, to generalize {\sc Cumulative Subtraction} further. Our model already allows for randomized updates of cumulations, via expected values, as may be convenient for some learning model, although we skip the details of a formalization here.

In mainstream Game Theory, each player aims to maximize their own utility. We define the {\em self-interest utility} of each Player~$i$ in a terminal state $(x,(C_1,C_2),p)^\T$ as equal to $C_i$ (more general utility functions will be introduced in Section~\ref{sec:cumgame}). 
Therefore, every grounded position $(x,(C_1,C_2),p)$ together with a utility function, induces an Extensive Form Game. To decide on the game solution, we apply the \emph{pure subgame perfect equilibrium} (PSPE) solution concept. We provide a detailed definition in Section~\ref{sec:extform}, but for now it is sufficient to recall that PSPE means that every player selects an action that maximizes their  utility, in every subtree of the game. See Figure~\ref{fig:ext_23} for a self-interest interpretation of {\sc Cumulative Subtraction}. The PSPE has unique utilities because we assume that players have well defined preferences for the other players. For example, in case of two players, they may be  (i) \emph{antagonistic}, that is, in facing equal utility, they aim to minimize the utility of the opponent, or they may be (ii) \emph{friendly}, and this is when they maximize the other player's utility in case of indifference. This is the first variant of a combinatorial game (that we are aware of) that is not zero-sum. In case of zero-sum games, instead of PSPE,  we will refer to the simpler, but equivalent notion of {\em optimal play}.
%\rmr{the PSPE is not unique, only the utilities. Consider e.g. a game with $x=3$ where $\s(3)=\{1,2\}$ and $\s(2)=\s(1)=\{1\}$. Then the player at $x=3$ is indifferent between taking 1 or 2, and both lead to a PSPE (with 3 and 2 moves, respectively). I guess this also occurs with Normal play, where a player may have more than one winning strategy.}

\subsection{Cumulative Subtraction's outcome functions}
A key concept in CGT is the \emph{outcome function}. Here we will generalize this concept to include self-interest play. In this section, we restrict attention to one heap games. %The outcome function depends on the ruleset, and assigns an `optimal play value' to every (ungrounded) game position. That is, it operates on layers 1 and 2 of the game, and had not yet any obvious relevance to (typically Layer~3) EGT contexts. \rmr{see my first comment in this section} 
%Notably, in symmetric zero-sum games, the outcome function captures the advantage to the starting player in every position (but it is indifferent to whether this is Player~1 or Player~2). 

Recall Property~\ref{prope:sym} ``symmetric rules''.  

\begin{observation} 
In this section, rules are symmetric if the subtraction map is symmetric, that is, if, for all $x$, $\s(x)_1=\s(x)_2$. In this case, we abuse notation and drop the set-index, i.e. we redefine the subtraction map as $\s:\N_0\rightarrow 2^\N$.
\end{observation}

The outcome function of the symmetric zero-sum variation was recently defined, for fixed rulesets, in \cite{CLMW}.\footnote{Outcome functions have been defined for other scoring games under various names in the literature, e.g.~\cite{M, H, E, Si, LNS, LNNS}.} We can generalize to a {\em heap-size dynamic} variation, without further ado. It is particularly appealing in its concise one-line definition. 

\begin{definition}[Outcome Symmetric Zero-sum \cite{CLMW}]\label{def:outS}
Consider a symmetric ruleset $\s$ in {\sc Zero-sum  Cumulative Subtraction}. The outcome of a heap of size $x$ is $o_\zs\left(x^\T\right) = 0$ if $\s\left(x^\T\right)=\emptyset$, and otherwise 
\begin{align}\label{eq:}
o_\zs(x) = \max \{a-o_\zs(x-a)\mid a\in \s(x)\}.
\end{align}
\end{definition}
%Note that the outcome is . 

This definition is a part of the inspiration for this paper. It is a formalization of the mentioned ``tabular approach'' in this setting (see Example~\ref{ex:cumsub}). In \cite{CLMW} it is shown that this outcome is always nonnegative for fixed subtraction sets. This need not be the case for heap-size dynamic games. If, for example  $\s(2)=2$ and $\s(3)=1$, then $o(3)=-1$. Such things cannot happen for fixed rulesets. Of course, the point of this defineition, is the following result \cite{CLMW}. 

\begin{proposition}[Symmetric, Zero-sum, \cite{CLMW}]\label{prop:zsutility}
Consider a symmetric ruleset $\s$ in {\sc Zero-sum  Cumulative Subtraction}. At any grounded position \newline $(x,(0,0),p)$, suppose that $\left(x^\T,\left(C_1^\T,C_2^\T\right),p'\right)$ is a terminal position under optimal play. Then, the utility is $o_\zs(x)=C_1^\T-C_2^\T$. 
\end{proposition}
\begin{proof}
This is immediate by Definition~\ref{def:outS}.
\end{proof}
Although the rules of game are more natural for self-interest than zero-sum ditto, the self-interest outcome function will require some extra consideration, namely the requirement of a tie-break rule for each player, in case of indifference. 

The outcome at every position $x$ should specify a value for both players, and this pair depends on who starts. Thus, in the general case, the {\em outcome} requires a $2\times 2$ matrix (that is generalized to an $n\times n$ matrix in Section~\ref{sec:cumgame}). For now we can use a simpler notation. 

The outcome is defined with respect to Player~1 as starting player, and by symmetry, the  outcome when Player~2 starts is obtained by reversing the player utilities. To simplify notation,  for symmetric rulesets we write $\s(x)=\s(x)_1=\s(x)_2$.

%\rmr{You also wanted to rewrite Def.3 for the 2X2 matrix. We can remove the "symmetry" assumption from the definition, and add an asymmetric example (say, like example 8 with self-interest). Proposition 1 still works, right?}
 
\begin{definition}[Outcome Symmetric Self-interest  Antagonistic]\label{def:outCSgs}
Consider a ruleset $\s$ in {\sc Self-interest  Cumulative Subtraction}, under antagonistic play. The symmetric self-interest outcome of a heap of size $x$ is the pair $ o_\si(x) = ({o_\si^1}(x), {o_\si^2}(x))$, where $o_\si(x)=(0,0)$ if $\s(x)=\emptyset$, and otherwise 
\begin{align}
{o_\si^1}(x)&=\max\{{o_\si^2}(x-a)+a\mid a\in \s(x)\}, \label{eq:max}\\
{o^2_\si}(x)&={o_\si^1}(x-a^{\!*}), \label{eq:min}
\end{align}
where
\begin{equation}\label{eq:argmin}
 a^{\!*} = \argmin_{a\in \s^*(x)}o_\si^1(x-a),
\end{equation}  
and where  $\s^*(x)\subseteq \s(x)$ denotes Player~1's set of indifference actions from~\eqref{eq:max}.
%The symmetric self-interest \emph{outcome} of heap size $x$ is the pair $ o_\si(x) = ({o_\si^1}(x), {o_\si^2}(x))$. 
\end{definition}
Let us recall the intuition of the indifference set; in general there can be many options that result in the same PSPE-value for the current player. All these options belong to the indifference set. We cannot assume apriori that the other player's utility is non-sensitive to the current players choice. For uniqueness of the computed values, a generic preference is declared before a game starts. 

When play is self-interest {\em friendly} in case of indifference, instead of $\argmin$ in \eqref{eq:argmin}, we choose $\argmax$, and otherwise the definition is the same. See Observation~\ref{obs:zssi} in Section~\ref{sec:discussion}, for discussions and conjectures on these variations of games and outcome functions.\footnote{Since the first version of this manuscript, Definition 3 has been generalized to allow any combination of antagonistic and friendly players \cite{BKLM}. Using this framework, the authors prove a monotonicity result for two-move games, namely that players cannot lose by acting friendly in cases of indifference.}

%\footnote{Since the first version of this manuscript, Definition~\ref{def:outCSgs} has been generalized to include any combination of antagonistic vs. friendly players; by using this the authors  prove a monotonicity result for two-move games \cite{BKLM}, namely that players cannot lose by acting friendly in cases of indifference.}

\begin{example}\label{ex:cumsub}
Consider the fixed symmetric subtraction set $\s=\{2,3\}$. An extensive form representation, playing from $x=7$, is in Figure~\ref{fig:ext_23}. All heap sizes and the moves between them (i.e., layer~1 up to $x=7$) are in Figure~\ref{fig:subtr23}(a); this DAG does not include cumulations. The computed outcome function under both zero-sum and self-interest variants is displayed in Figures~\ref{fig:subtr23}(b) and (c), respectively.

Note that the representation is smaller than the na\"ive representation in Figure~\ref{fig:ext_23}, and yet it still allows for computing the PSPE. The outcome in both variations is realized via the play sequence $7\mapsto5\mapsto 2\mapsto 0$. 

Both these representations are nice, but due to the Layers~1 and~2, it is much more convenient to present the outcomes in an easily extendable table, similar to the ones for normal play. Next, we display the zero-sum and self-interest outcomes for the fixed symmetric subtraction game $\s=\{2,3\}$. It turns out that, for this ruleset, the self-interest outcomes do not depend on the choice of tie-breaking rule.

\begin{center}
\begin{tabular}{|c| c c c c c c c c c c|}
\hline
$x$		&0 &1 &2 &3 &4 &5 &6 &7  & 8&9 \\ 
\hline
$o_\zs(x)$	&0 &0 &2 &3 &3 &1 & 0 & 1 & 2&3 \\ 
\hline
$o_\si(x)$	&(0,0) &(0,0) &(2,0) &(3,0) &(3,0) &(3,2)  &(3,3) &(4,3) & (5,3)& (6,3) \\ 
\hline
\end{tabular}
\end{center}
For example $o_\zs(6)=\max\{-o_\zs(3)+3,-o_\zs(4)+2\}=0$, and similarly $o_\si(6)=\overline{o}_\si(4)+(3,0)=(3,3)$, if $\overline o$ indicates swapped entries. Observe that the entry for $x=7$, self-interest coincides with the PSPE-computation in Figure~\ref{fig:ext_23}. %This can also be useful when the game tree/DAG becomes complicated (exponential/quadratic growth compared to linear). In fact, these type of `outcome tables' inspired this work; we wanted to explore how far we can extend the games, while still having a `one-line table computation' of outcomes. And note that these type of tables prints all outcomes in the spirit of layers 1 and 2 in the game representation, not just the ones that belong to a particular grounded (Layer~3) game, as in the DAG representation or, as is the case in the extensive form representation, where  the tree only concerns computation of the value at the root, for a given starting player. For example, if we later want to display further outcomes, we will require $o(6)$, which is mossing in the tree representations.
\end{example}
%\tikzstyle{vertex}=[circle,fill=black!20, minimum size=40pt, inner sep=0pt]
\tikzstyle{edge} = [draw,thick,-]
\tikzstyle{weight} = [font=\small]
\begin{figure}[t]
\begin{tikzpicture}[scale=.85, auto,swap]
\tikzstyle{vertex}=[circle, fill=gray!10, minimum size=38pt, inner sep=0pt]
    % Draw the vertices
    %Outcome 0
    \begin{small}
        \foreach \pos/\name in { {(0,7)/S}}
        \node[vertex] (\name) at \pos {$7;(0,0)$};    
        \foreach \pos/\name in {{(-2,5)/B}}
        \node[vertex] (\name) at \pos {$5;(2,0)$};
        \foreach \pos/\name in {{(-4,3)/BB}}
        \node[vertex] (\name) at \pos {$3;(2,2)$};
        \foreach \pos/\name in { {(-6,1)/BBB}}
        \node[vertex] (\name) at \pos {$1;(4,2)$};
        \foreach \pos/\name in { {(-4,1)/BBC}}
        \node[vertex] (\name) at \pos {$0;(5,2)$};
         
         \foreach \pos/\name in {{(-1,3)/BC}}
        \node[vertex] (\name) at \pos {$2;(2,3)$};
        %Outcome 2
        \foreach \pos/\name in {  {(-2,1)/BCB}}
        \node[vertex] (\name) at \pos {$0;(4,3)$};
        \foreach \pos/\name in {  {(2,5)/C}}
        \node[vertex] (\name) at \pos {$4;(3,0)$};
 `````%Outcome 3
        \foreach \pos/\name in { {(1,3)/CB}}
        \node[vertex] (\name) at \pos {$2;(3,2)$};
        \foreach \pos/\name in { {(4,3)/CC}}
        \node[vertex] (\name) at \pos {$1;(3,3)$};
    \foreach \pos/\name in { {(0,1)/CBB}}
        \node[vertex] (\name) at \pos {$0;(5,2)$};
           
    % Connect vertices with edges and draw rewards
    \foreach \source/ \dest /\weight in {S/B/2, S/C/3,
                                         B/BB/2, B/BC/3,
                                         BB/BBB/2,BB/BBC/3,
                                         BC/BCB/2,C/CB/2, C/CC/3, CB/CBB/2
                                        }
\path[edge,thin] (\source) -- node[weight] {$\weight$} (\dest);
    \foreach \source/ \dest /\weight in {S/B/2, 
                                         B/BC/3,
                                         BB/BBC/3,
                                         BC/BCB/2,C/CC/3, CB/CBB/2
                                        }
\path[edge,double] (\source) -- node[weight] {$\weight$} (\dest);
\node at (-8.5,7) {$p=1$};
\node at (-8.5,5) {$p=2$};
\node at (-8.5,3) {$p=1$};
%\node at (-8.5,1) {$\tilde p=2$};
\end{small}
\end{tikzpicture}
\caption{
This is a reiteration of Figure~\ref{fig:ext_23-squirrel}, to exemplify a PSPE analysis: the Extensive Form Game induced by the self-interest symmetric fixed subtraction set $\s=\{2,3\}$ {\sc Squirrel Pebbles} and the grounded position $(7,(0,0),2)$ (where Player~1 starts). In each node we show $x;(C_1,C_2)$, while the current player is shown to the left. The cumulations at the leaf levels are also the utilities of the players. The thick edges mark the PSPE of this game. \label{fig:ext_23}}
\end{figure}

\tikzstyle{vertex}=[circle,fill=green!15, minimum size=20pt, inner sep=0pt]
\tikzstyle{edge} = [draw,thick,-]
\tikzstyle{weight} = [font=\small]
\begin{figure}
\centering
\begin{tikzpicture}[scale=0.65, auto, swap]
    \node at (0,-1) {(a) Heap sizes};
    % Draw the vertices
    \foreach \pos/\name in {{(0,7)/7}, {(-1,5)/5}, {(1,5)/4},
                            {(-2,3)/3}, {(0,3)/2}, {(2,3)/1}, {(-1,1)/0}}
        \node[vertex] (\name) at \pos {$\name$};
        \foreach \pos/\name in { {(-3,1)/A}}
        \node[vertex] (\name) at \pos {$1$};
    % Connect vertices with edges and draw rewards
    \foreach \source/ \dest /\weight in {5/7/2, 4/7/3,
                                         2/5/3, 3/5/2,
                                         2/4/2,1/4/3,
                                         A/3/2,0/3/3,
                                         0/2/2}
\path[edge] (\source) -- node[weight] {$\weight$} (\dest);
\end{tikzpicture}
\hspace{3 mm}
\begin{tikzpicture}[scale=0.65, auto,swap]
    \node at (0,-0.5) {(b) CGT outcome};
    \node at (0,-1.3) {(0-sum)};
\tikzstyle{vertex}=[circle,fill=gray!10, minimum size=20pt, inner sep=0pt]
    % Draw the vertices
    %Outcome 0
    \foreach \pos/\name in {{(2,3)/A}, {(-3,1)/B}, {(-1,1)/C}}
        \node[vertex] (\name) at \pos {$0$};
        %Outcome 1
        \foreach \pos/\name in { {(0,7)/D}, {(-1,5)/E}}
        \node[vertex] (\name) at \pos {$1$};
        %Outcome 2
        \foreach \pos/\name in {  {(0,3)/F}}
        \node[vertex] (\name) at \pos {$2$};
 `````%Outcome 3
        \foreach \pos/\name in { {(1,5)/G},{(-2,3)/H}}
        \node[vertex] (\name) at \pos {$3$};
    % Connect vertices with edges and draw rewards
    \foreach \source/ \dest /\weight in {E/D/2, G/D/3,
                                         F/E/3, H/E/2,
                                         F/4/2,A/G/3,
                                         B/H/2,C/H/3,
                                         C/F/2}
\path[edge] (\source) -- node[weight] {$\weight$} (\dest);
  \foreach \source/ \dest /\weight in {E/D/2,
                                         F/E/3, 
                                         A/G/3,
                                         C/H/3,  C/F/2
                                         }
\path[edge,double] (\source) -- node[weight] {$\weight$} (\dest);
\end{tikzpicture}
\hspace{3 mm}
\begin{tikzpicture}[scale=0.65, auto,swap]
\tikzstyle{vertex}=[circle,fill=green!15, minimum size=20pt, inner sep=0pt]
    % Draw the vertices
    %Outcome 0
    \node at (0,-0.5) {(c) CGT outcome}; \node at (0,-1.3) {(self-interest)};
    \foreach \pos/\name in {{(2,3)/A}, {(-3,1)/B}, {(-1,1)/C}}
        \node[vertex] (\name) at \pos {$(0,0)$};
        %Outcome 1
        \foreach \pos/\name in { {(0,7)/D}}
        \node[vertex] (\name) at \pos {$(4,3)$};
         \foreach \pos/\name in {{(-1,5)/E}}
        \node[vertex] (\name) at \pos {$(3,2)$};
        %Outcome 2
        \foreach \pos/\name in {  {(0,3)/F}}
        \node[vertex] (\name) at \pos {$(2,0)$};
 `````%Outcome 3
        \foreach \pos/\name in { {(1,5)/G}}
        \node[vertex] (\name) at \pos {$(3,0)$};
           \foreach \pos/\name in {{(-2,3)/H}}
        \node[vertex] (\name) at \pos {$(3,0)$};
    % Connect vertices with edges and draw rewards
    \foreach \source/ \dest /\weight in {E/D/2, G/D/3,
                                         F/E/3, H/E/2,
                                         F/4/2,A/G/3,
                                         B/H/2,C/H/3,
                                         C/F/2}
\path[edge] (\source) -- node[weight] {$\weight$} (\dest);
    \foreach \source/ \dest /\weight in {E/D/2,
                                         F/E/3, 
                                         A/G/3,
                                         C/H/3,  C/F/2
                                         }
\path[edge,double] (\source) -- node[weight] {$\weight$} (\dest);
\end{tikzpicture}

\caption{To the left: symmetric {\sc Cumulative Subtraction} with $S=\{2,3\}$ represented as a DAG played from a heap of size 7. The middle picture displays  the zero-sum outcomes, and to the right, the general-sum outcomes. The double edges indicate which child determines the parent's value. %See also Figure~\ref{fig:extf23}.
}\label{fig:subtr23}
\end{figure}

The general (not necessarily symmetric) zero-sum outcome function has a very similar expression as Definition~\ref{def:outS}. The only difference is that the outcome function now requires a move flag corresponding to the previous player index of the ruleset and we may no longer assume that Player~1 starts. 

\begin{definition}[Outcome Zero-sum]\label{def:outSgen}
Given a ruleset $\s$ in {\sc Zero-sum Cumulative Subtraction}, the outcome of a heap of size $x$ is $o_\zs(x)=(o_\zs(x,1),o_\zs(x,2))$, where $$o_\zs(x,1)=\max(\{o_\zs(x-a,2)+a\mid a\in \s(x)_1\})$$ and $$o_\zs(x,2)=\min(\{o_\zs(x-a,1)-a\mid a\in \s(x)_2\}).$$
\end{definition}

Next, we define the, not necessarily symmetric, self-interest variation (in antagonistic play). Again, although the self-interest utility function is a more natural model in  situations where the players do not directly compete eachother, the possiblity of indifference situations requires extra care. The outcome is succinctly represented by a $2\times 2$ matrix, because, similar to Definition~\ref{def:outSgen}, it is important to distinguish `who is the current player?'. 

\begin{definition}[Outcome Self-interest  Antagonistic]\label{def:outCSgsPart}
Consider an instance of {\sc Self-interest Cumulative Subtraction} under antagonistic play. The self-interest outcome of a heap of size $x$ is the matrix 
\begin{equation*}
 o_\si(x) =\begin{bmatrix}
o_\si^1(x,1) & o_\si^2(x,1)\\
o_\si^1(x,2) & o_\si^2(x,2),
\end{bmatrix} 
\end{equation*}
%The $i^{\rm th}$ row of this matrix concerns Player~$i$ as starting player, and The $j^{\rm th}$ column concerns Player~$j$'s game value, depending on who starts. 
where $o_\si^i(x,p)$ denotes Player~$i$'s grounded game value when player $p$ starts. 
Here $o_\si(x,p)=(0,0)$ if $\s(x)_{p}=\emptyset$, and otherwise 
\begin{align}
{o_\si^1}(x,1)&=\max\{{o_\si^1}(x-a,2)+a\mid a\in \s(x)_1\}, \label{eq:max1}\\
{o^2_\si}(x,1)&={o_\si^2}(x-a^{\!*},2), \label{eq:min}
\end{align}
where
\begin{equation}\label{eq:argmin1}
 a^{\!*} = \argmin_{a\in \s^*(x)_1}o_\si^1(x-a,2),
\end{equation}  
and where  $\s^*(x)_1\subseteq \s(x)_1$ denotes Player~1's set of indifference actions from~\eqref{eq:max1}; 
\begin{align}
{o_\si^2}(x,2)&=\max\{{o_\si^2}(x-a,1)+a\mid a\in \s(x)_2\}, \label{eq:max2}\\
{o^1_\si}(x,2)&={o_\si^1}(x-a^{\!*},1), \label{eq:min}
\end{align}
where
\begin{equation}\label{eq:argmin2}
 a^{\!*} = \argmin_{a\in \s^*(x)_2}o_\si^2(x-a,1),
\end{equation}  
and where  $\s^*(x)_2\subseteq \s(x)_2$ denotes Player~2's set of indifference actions from~\eqref{eq:max2}. 
\end{definition}
There is some abuse of terminology. Equations~\eqref{eq:argmin1} and \eqref{eq:argmin2} may have several solutions; but then they all produce the same result.
\begin{example}\label{ex:cumsub2}
Let us display the initial zero-sum and self-interest outcomes, for the fixed Cumulative Subtraction game $\s=(\{2,3\},\{1,4\})$:
%\vspace{-3mm}
\begin{center}
\begin{tabular}{|c|c| c c c c c c c c |}
\hline
$\s$ & $\ \ \ \ x$		&0 &1 &2 &3 &4 &5 &6 &7  \\ 
\hline
$\{2,3\}$ & $o_\zs(x,1)$	&0 &0 &2 &3 &2 &3 & 4 & $-1$  \\ 
$\{1,4\}$ & $o_\zs(x,2)$	&0 &$-1$ &$-1$ &1 &$-4$ &$-4$ & $-2$ & $-1$  \\ 
\hline
$\{2,3\}$&$o_\si(x,1)$	&$(0,0)$ &$(0,0)$ &$(2,0)$ &$(3,0)$ &$(3,1)$ &$(4,1)$  &$(5,1)$ &$(3,4)$  \\ 
$\{1,4\}$ & $o_\si(x,2)$	&$(0,0)$ &$(0,1)$ &$(0,1)$ &$(2,1)$ &$(0,4)$ &$(0,4)$ &$(2,4)$ &$(3,4)$  \\ 
\hline
\end{tabular}\vspace{1 mm}
\begin{tabular}{|c|c| c c c c c c c c |}
\hline
$\s$ & $\ \ \ \ x$		&7 &8 &9 &10 &11 &12 &13 &14  \\ 
\hline
$\{2,3\}$ & $o_\zs(x,1)$	&$-1$ &0 & 1& 2& 0& 2&3  & $-2$  \\ 
$\{1,4\}$ & $o_\zs(x,2)$	& $-1$&$-2$ &$-1$ & 0&$-5$ &$-4$ & $-3$ & $-2$  \\ 
\hline
$\{2,3\}$&$o_\si(x,1)$	&$(3,4)$ &$(4,4)$ &$(5,4)$ &$(6,4)$ &$(6,6)$ &$(7,5)$  &$(8,5)$ &$(6,8)$  \\ 
$\{1,4\}$ & $o_\si(x,2)$	&$(3,4)$ &$(3,5)$ &$(4,5)$ &$(5,5)$ &$(3,8)$ &$(4,8)$ &$(5,8)$ &$(6,8)$  \\ 
\hline
\end{tabular}
\end{center}
That is, Player~1 subtracts 2 or 3, whereas Player~2 subtracts 1 or 4. So far, the self-interest outcomes do not depend on the particular tie-breaking rule (becasue so far, in all cases of ties, Equations~\eqref{eq:argmin1} and \eqref{eq:argmin2} have two solutions). %Let us make some immediate reflections on these initial outcomes. And indeed, 
\end{example}
Let us present an intermediate step towards general {\sc Cumulative Games}.
\begin{proposition}[Self-interest Cumulative Subtraction]\label{prop:SG_self}
 Fix a current player $p\in\{1,2\}$. At any grounded position $(x,(C^0_1,C^0_2),2)$, let $(x^T,(C^T_1,C^T_2))$ be the final position under PSPE play. Then, ${o_\si}(x,p)=(C^T_1-C^0_1,C^T_2-C^0_2)$.  That is, the utilities in PSPE play are $C^T_1-C^0_1$ and $C^T_2-C^0_2$ for Player~1 and~2 respectively.
\end{proposition}

This will be proved in a more general setting in Section~\ref{sec:cumgame}, in Theorem~\ref{thm:main1} and Theorem~\ref{thm:main2}, the latter which concerns  {\em Self-interest Utililty Heap Dynamic Cumulative Games}; this ruleset restriction will simply be referred to as {\sc Heap Dynamic}. It generalizes all games in this section.

By  Proposition~\ref{prop:SG_self}, the rows in the table represent game values of the corresponding grounded positions, whenever $(C_1^0,C_2^0)=(0,0)$. For the zero-sum variation it is noted in \cite{CLMW} that all symmetric outcomes are nonnegative, but this will no longer hold in the non-symmetric case (for example $o_\zs(3,2)=1$ and $o_\zs(7,1)=-1$), and moreover, now a player may get a zero-sum outcome larger than the maximum of their subtraction set, e.g. $x=6$, which is not possible in case of symmetry. We postpone study of combinatorial properties and asymptotic of non-symmetric games, similar to and generalizing \cite{CLMW}. %, and/or encourage other researchers to take on the many interesting topics, emerging from the definitions in this section. 
In Section~\ref{sec:discussion}, we guide the reader into some interesting problems related to this section. 
%In this study, we de-emphasize such combinatorial aspects of solutions; instead, we wish to convey that for generalized optimal/PSPE play evaluations, sometimes the convenient means is recurrence via the outcome function. 
Before we move on, we review the complexity of {\sc Cumulative Subtraction}.

\subsection{Layer~3 complexity}
In general Extensive Form Games, computing a PSPE requires traversing the entire game tree by backward induction. In the worst case, the size of the game tree may be exponential in its height. In the case of {\sc Cumulative Subtraction}, we can merge identical nodes (e.g. in Figure~\ref{fig:ext_23} node $[0;(5,2)]$ appears twice). Since in a tree rooted in $x$, with a fixed initial cumulation, there are at most $x$ different cumulation values for each player, we can construct a DAG representation with  $x^2$ nodes, which gives us an upper bound of $O(x^2)$ on the computation time of a PSPE, for every game starting from a particular grounded position.%\footnote{In general, we do not want to fix an initial cumulation, since we are interested in the outcomes on all heap sizes; the concept of initial cumulation belongs on Layer~3, rather than the CGT-type layers 1 and 2.} 

In some cases the outcome function allows us to compute the PSPEs even more efficiently. 

\begin{proposition}\label{prop:SG_complex}
Consider any ruleset $\s$ with $\sup_x |\s(x)| \leq M$ for some constant $M$. 
Computing all outcomes $o_\si(x,p)$ for all $x\leq \hat x$, can be done in time $O(\hat x)$. In particular, a PSPE of any Extensive Form Game grounded in $(\hat x,(C^0_1,C^0_2),p)$ can be computed in time $O(\hat x)$.
\end{proposition}
\begin{proof}
The first part follows directly from Definition~\ref{def:outCSgs}, as we can compute the outcome function via dynamic programming. 

For PSPE computation from a particular grounded position, we compute the outcome $o_\si(x,p)=(C_1,C_2)$. By Proposition~\ref{prop:SG_self}, we get the terminal utilities of both players as $u_1=C_1+C^0_1$ and $u_2=C_2+C^0_2$.
\end{proof}

\section{Cumulative Games}\label{sec:cumgame}
In Section~\ref{sec:cumsub} we considered a specific class of self-interest combinatorial games based on collecting pebbles, and we showed that  equilibria are efficiently computed, as expressed via the CGT outcome function. 

Our purpose in this section is to define a more general class of games and extend Proposition~\ref{prop:SG_self}. %(Later, in Section~\ref{sec:equiv},  we show it is general enough to capture all Extensive Form Games.)
We first lay out the general definitions of {\sc Cumulative Games}, and in particular we generalize the outcome function. 
It is quite obvious that it is not possible to compute PSPE in every game in time that is subexponential in height, which means that a naive  generalization of Propositions~\ref{prop:SG_self} and \ref{prop:SG_complex} could not apply to all games. 
Indeed, in Example~\ref{ex:auc}, we follow up on the auction game example in the Prologue, for which the proposition does not continue to hold, even after a remodeling attempt.

This section is in preparation for Section~\ref{sec:extform}, where we develop the required tools to reason about strategy profiles (in the EGT sense) for general {\sc Cumulative Game}s, and in Section~\ref{sec:valout}, where we define a general variant of the outcome function (which attributes an $n\times n$ matrix of values to every position). We  identify a property of {\sc Cumulative Games}, that is a sufficient condition for a generalization of Proposition~\ref{prop:SG_self}; see Theorem~\ref{thm:main1}.

\subsection{Cumulative Game Form}\label{sec:CGF}
Let us begin giving the definitions for {\sc Cumulative Games}. The general setting involves $n\ge 2$ players and $d\in\N$ heaps, and it begins with the notion of a Cumulative Game Form, and its three layers. %(we can think about each heap as a different type of pebbles)
  Later, we apply the utilities (Definition~\ref{def:util}), and a generalized CGT-type {\em outcome function} (Section~\ref{sec:valout}). 
    
%  Layer~1 contains all possible (starting) positions, and Layer~2 is an instance of Layer~1 together with a ruleset. Layer~2 is adapted for disjunctive sum play, and as we will see, it belongs to a `PSPE-outcome class'. Layer~3 includes the information about the previous player, together with a turn function, and it defines the current player. Together with the player utilities, this layer thus defines an Extensive Form Game. 
  
\begin{definition}[Cumulative Game Form]\label{def:cumgameform}
An $n$-player Cumulative Game Form on $d$ heaps of pebbles is  a 5-tuple $F=(n,d,R,\Omega,p)$. %is a ruleset $R$ defined on a heap space $\Omega$, together with a turn function  
There are three layers in a game form:
\begin{description}
\item[Layer~1.] The heap space is $\Omega=(\N_0\times \R^{n})^d$;
\item[Layer~2.] A heap position is $\om\in \Omega$, together with a ruleset $R= (\A,\re)$;% feasible ruleset $R$;  
\item[Layer~3.] A grounded position $(\om,p)\in \Omega\times[n]$, is a heap position, or a disjunctive sum of heap positions, together with a current player $p$, and a turn function $\gamma:\Omega\times[n]\rightarrow [n]$.  %\footnote{As a shorthand, when $\om$ is understood (or irrelevant) we may write $\gamma(\om,p)=\p$.} 
%A grounded position can also be a disjunctive sum of heap positions, together with a heap position, a  current player and a turn function.
%\item[Layer 4.] An $n$-tuple of identical heap positions, $(\om,\ldots ,\om)$, one for each previous player.
\end{description}
The current player $p$ is given by the heap position and the previous player $p'$, as prescribed by the turn function, $p=\gamma(\om,p' )$. The idea is that a grounded position is a recursive construction (similar to Layer~2), where all subsequent positions will also be grounded.

A typical heap position $\om\in \Omega$ is $\om=((x_1, \C_1),\ldots , (x_d, \C_d))$. 
We may regard the cumulation $\C$ as an $n\times d$ real matrix,  where the entry at $(i,h)$ is Player~$i$'s cumulation on heap $h$. Row $i$ represents the heap cumulations for Player~$i$, whereas column $h$ represents the  player cumulations on heap $h$. If, given $\om$, we wish to extract the cumulation matrix, then we write $\C=\C(\om)$. In the recursive construction, \eqref{eq:gitoo}, it is convenient to abuse notation and write instead $\om=(\x,\C)$, where: 
\begin{itemize}
\item[(i)] $\vec x = (x_1,\ldots , x_d)\in \N_0^d$ is a $d$-tuple of heaps. 
\item[(ii)] the cumulation vector, for each heap $h\in[d]$, is $\C_h\in \R^n$.
\end{itemize}

The ruleset $R$ is described by: for all $\om\in\Omega$, with $\om = (\x,\C)$, for each player $i\in\n$, the option set is
\begin{align}\label{eq:gitoo}
\om^i = \{(\x+\a,\C+\re(\om,i,\a))\mid \a\in \A(\vec\omega)_i\},
\end{align}
%where $\om^{(\a)}=(\x+\a,\C+\re(\om,p,\a))$. 
 where, 
\begin{itemize}
\item[(iii)] the {\em action-map} $\A:\Omega \rightarrow \left(2^{\Z^d}\right)^n$ specifies, for each heap position, the set of allowed actions on the heaps for each player (possibly some set is empty);
\item[(iv)] the {\em reward} $\re: \Omega \times \n\times \Z^d \rightarrow \R^{n,d}$ depends on a heap position $\om$, Player~$i$ and their action $\a\in \A(\om)_i$. 
\end{itemize}
Regarding item (iii), if $\A(\om)_p=\varnothing$, then $\om$ is a terminal grounded position, with respect to the current player $p$ and cyclic move order say. This does not imply that $\om$ is a `terminal position' at Layer~2 because some player $-p$ might have a non-empty action set (in fact this latter notion is not defined unless it holds for all players). %Here the ruleset is the pair $R = (\A,\re)$.

%A grounded position $(\om,p)$ is \emph{terminal} if $\A(\om,p)=\emptyset$, and otherwise, a grounded position $(\om,p)$ is identified with its set of options:
%for all players $p\in\n$, for all $\om\in\Omega$, with $\om = (\x,\C)$, 
%\begin{align}\label{eq:gitoo}
%(\om, p) = \{(\om^{(\a)}, \gamma(\om, p))\mid \a\in \A\},
%\end{align}
%where $\om^{(\a)}=(\x+\a,\C+\re(\om,p,\a))$. 
\end{definition}

%We will sometimes refer to the set of all heap positions, without cumulations, as $X = X(\Omega)$; see Section~\ref{sec:valout}.

A CGF does not yet specify player  utilities. The utilities are only realized at terminal positions, i.e. when the current player cannot move on any heap. Therefore, all CGFs will be assumed feasible. 

As with Extensive Form Games, the ``game form'' alone dictates how the game \emph{can} be played but not \emph{what} players should do. 
Intuitively, players want to maximize their individual utilities, where the utility of each player depends on the particular terminal grounded position $(\om^\T,p)$, i.e., the current player $p$ cannot move. We omit the superscript `T' as in terminal in the definition of the utility map. In some extreme situation, every player except the current player has plenty move options. But the game ends if the current player cannot move. At that point, the `winning condition' (or utility) is invoked.

\begin{definition}[Utility Map]\label{def:util}
Consider a CGF on a feasible ruleset. 
A \emph{utility map} for Player~$i$ on heap $h$, is a function $u_{i,h}:\R^n\times \n \rightarrow \R^n$, which maps a terminal cumulation on heap $h$ to Player~$i$'s utility. Player~$i$'s utility at any terminal grounded position $(\om,p)$ is 
\begin{align}\label{eq:gitoou}
u_i(\C,p)=\sum_{h\in [d]}u_{i,h}\left(C_{i,h},\; p\right), 
\end{align}
where $\C=\C(\om)$. Let $\ut(\C,p)=(u_1(\C,p)),\ldots , u_n(\C, p))$.
\end{definition}

\begin{observation} 
The utility functions are sensitive to `who is to play?', and this is useful, for example when utility should simulate normal play ending, on all heaps.\footnote{We have even more motivation in the normal play embedded Guaranteed Scoring play \cite{LNS, LNNS} and Absolute CGT \cite{LNSabs1,LNSabs2}.} The reward function is the correct means, whenever we wish to simulate normal (or mis\`ere) play on individual heaps. %See also Example~\ref{ex:contprel}.
\end{observation}

\begin{definition}[{\sc Cumulative Game}]\label{def:cumgame}
A  Cumulative Game Form $F$ together with a utility map $\ut$ induces a {\sc Cumulative Game} $(F,\ut)$. 
\end{definition}

%For a given current player, we may want to be specific, and refer to a grounded {\sc Cumulative Game}, etc.

Note that every grounded {\sc Cumulative Game} is an Extensive Form Game (see Section~\ref{sec:equiv}), and so we can talk about specific strategy profiles etc.

Recall two-player {\sc Cumulative Subtraction}.  
\begin{observation}\label{ex:cumsub_special}
Both {\sc Self-interest} and {\sc Zero-sum Cumulative Subtraction} from Section~\ref{sec:cumsub}, are simple special cases of a {\sc Cumulative Game}. Both of them have the same Cumulative Game Form, where there are two players, $n=2$; a single heap: $d=1,\Omega = \N \times \R^2$; the action map $\A$ is the negation of $\s$;  for each player, the  reward is the `identity map': for all $\vec \omega$, $\re(\om,1,a) = (a,0)$ and $\re(\om,2,a) = (0,a)$; and the turn function is alternating: for all $\vec \omega$, $\gamma(\om,p')=p$. In the self-interest variation, $u_i(\om,p)=C_i$, whereas in the zero-sum variation, $u_i(\om,p) = C_i-C_{-i}$. \end{observation}

In view of our results to come we will require utility to be the identity map, i.e. for all players $i$, $u_i(\C,p) = C_i$, a.k.a. self-interest (Property~\ref{prope:si}). But this was not the case in Definition~\ref{def:cumsub} for our motivational ruleset zero-sum symmetric {\sc Cumulative Subtraction}. The situation has a simple remedy, via the reward function. %In the {\sc Heap Dynamic} restriction to come,  we will not require rewards to be identity, but they must be cumulation independent. 
{\sc Zero-sum Cumulative Subtraction} does not have identity utility when stated as above, but we can define an equivalent game with identity utilities, by pushing the zero-sum interaction into the reward function. Let us model zero-sum {\sc Cumulative Subtraction} as a {\sc Cumulative Game} with self-interest utility. 

\begin{proposition}\label{prop:zsidut}
For all $\om$ with $a(\om)\in\A$, set $\re(\om, 1, a) = (a,-a)$ and $\re(\om,2,a) = (-a,a)$. If utilities are identity, i.e. $u_i(\C,p) = C_i$, then any sequence of play gives the same utility as {\sc Zero-sum Cumulative Subtraction}, and in particular, the optimal play strategies are identical. 
\end{proposition}
\begin{proof}
{\sc Zero-sum Cumulative Subtraction} is reviewed in Observation~\ref{ex:cumsub_special}, and the assertion follows.
\end{proof}
%Note that the cumulation part of a heap position is not used in modelling {\sc Cumulative Subtraction} as self-interest. This should be put into contrast with the motivating examples, Example~\ref{ex:auc} and (in particular) Example~\ref{ex:cumrew} in Section~\ref{sec:auc}. 

%Notice that rewards and actions, but not utility, is included in the notion of a ruleset. (Obviously utility depends on the (terminal) cumulation.) 
We have arrived at  motivating examples towards the modeling of Section~\ref{sec:valout} and the main results. Sometimes it is possible to model one and the same ruleset in different ways, by for example rearranging the rewards and utilities. For the purpose of the next property, we include the utility-map into the notion of a ruleset.
%Let us formalize two central concepts.

\begin{property}[Self-interest]\label{prope:si}
A ruleset  is \emph{self-interest} if the utility map is identity, with respect to terminal player cumulations. That is, for each player $i$ and each heap $h$, $u_{i,h}(C_{i,h},p)=C_{i,h}$, independent of the current player $p$.\end{property}

\begin{property}[Cumulation Independence]\label{prope:cumind}
A ruleset  is \emph{cumulation independent} if the actions, rewards and  the turn function do not depend on the cumulations. 
\end{property}

These two concepts interact in various ways, while modeling {\sc Cumulative Games} with various properties, and we will study some theoretical consequences of both. At first, we recall the sample ruleset {\sc Auction Pebbles} from Section~\ref{sec:intrex}.

\subsection{Auction Pebbles}\label{sec:auc}
Let us continue the discussion of the {\sc Auction Pebbles} variation of {\sc Cumulative Subtraction} from Section~\ref{sec:intrex}, where the PSPE strategy profile, but not the ruleset, depends on the cumulation vector. We will vary the utility function, rewards and the cumulation dependence, while modeling the same setting. 

In Examples~\ref{ex:auc} and \ref{ex:cumrew}, we illustrate that: for a  modest economic style generalization of the classical combinatorial game {\sc Cumulative Subtraction}, one should not hope for generic outcome representations via the tabular approach.%, as in the setting of our various examples on {\sc Subtraction Games} etc. 

However, we can still define an outcome function that is more efficient than a generic Extensive Form Game tree approach. %In either case, one can use the more expensive standard PSPE algorithm to compute the $n$ values that comprise the outcome. 
The second purpose hints at bridges between classical CGT {\sc Subtraction Games} to classical EGT English Auctions.

We review a variation of the one heap {\sc Auction Pebbles} by using the new notation. This example has a \emph{symmetric utility function}, and hence we write $u_i(\C)=u_i(\C,\cdot)$ for the utility of Player~$i$, assuming that Player~1 starts.

\begin{example}[{\sc Auction Pebbles}]\label{ex:auc}
Consider {\sc Cumulative Subtraction} as in Definition~\ref{def:cumsub}, with a fixed symmetric ruleset $\s=\{2,3\}$, announced by an auctioneer, and with identity rewards. The starting heap size is $x^0\in \N$, and Player~1 starts bidding on a single item of value $v\ge x^0$. Each player has offered an initial bid, corresponding to an initial cumulation vector $\C^0$.  During play they can increase their bids according to their actions: their bids increment with the sizes of their subtractions. 

For each player $i\in \{1,2\}$ the utility depends on the final cumulation vector $\C$, and it is of the form $u_i(\C)=v-C^i$ if $C^i> C^{-i}$, and otherwise $u_i(\C)=0$. That is, the players desire to win the auction, but with the smallest possible margin, because they pay their bid. (Similar to ``chicken game'' this game has an element of win-loss, but is still not a zero-sum game). 
 
From a heap of size $x^0=3$, with initial bid vector $\C^0=(0,0)$ or $\C^0=(0,1)$, Player~1 wins the item in a single turn by removing 2 pebbles, so this is equilibrium play. However, if the initial cumulation is $\C^0=(0,2)$, then, in equilibrium play, Player~1 must remove 3, and wins with a lower utility, if we set $v=4$. Thus, Player~1's PSPE strategy depends on the initial cumulation, although the rules do not depend on cumulations (!). 
\end{example}

The following example complements Example~\ref{ex:auc} in the sense that we can model the ruleset {\sc Auction Pebbles} as self-interest if we transfer the original function of the utilities to the rewards. However, as we will see, this slick maneuver will turn the ruleset cumulation dependent (!). %It is a main motivation for this study, a context that will be developed further in  Section~\ref{sec:valout}, where we prove main results. 

\begin{example}[Modified Utility {\sc Auction Pebbles}]\label{ex:cumrew}
Interestingly enough, Example~\ref{ex:auc} can be rewritten with identity utility, if we elaborate the rewards to mimic the situation. It is non-intuitive, but certainly doable for simple games like this. Let 
\begin{enumerate}
\item $r_1((3,(0,0)),1,2)=2$, 
\item $r_1((3,(0,0)),1,3)=1$, 
\item $r_1((3,(0,1)),1,2)=2$, 
\item $r_1((3,(0,2)),1,2)=0$,
\item $r_1((3,(0,2)),1,3)=1$,
\end{enumerate}
and so on. Player~1 will use exactly the same strategy as in the previous example if we set self-interest utility, i.e. $u_1(\C) = C_1$. This, however, requires that we let the reward depend, not only on the action taken (and perhaps the heap size), but also on the cumulation. This is the second condition, that we will disallow in a heap size dynamic computation of the outcome (in equilibrium). Note for example the distinction between items~3. and~4. In item~3. Player~1 wins and her utility is 2, whereas in item~4. nobody wins and her utility is 0. In the latter case it is beneficial to play instead as in item~5., which indeed coincides with Example~\ref{ex:auc}.
\end{example}

Of course, for generic {\sc Cumulative Games}, where rules may depend in a complicated way on player cumulations, we should not hope for more efficient computation of game values/outcomes than what is given by a generic non-efficient PSPE computation (typically exponential in the depth of the game tree). 

The `fix' of the PSPE cumulation dependency, via the reward function, to obtain identity utility, improves the situation with respect to the first main result, Theorem~\ref{thm:main1}, but not with respect to the second main rest, Theorem~\ref{thm:main2}. With respect to Theorem~\ref{thm:main1} the rule cumulation dependency is not a problem, while self-interest (identity utility) is the required property to have a utility related outcome function.   

Of course, we know already that sometimes this `fix' gives a neat outcome, via Proposition~\ref{prop:zsidut}. There is a fine distinguishing line somewhere between these examples and Proposition~\ref{prop:zsidut}, and to make a formal treatment of this ``fine line'', we will use definitions from main stream Game Theory via the Extensive Form Games. 

Without the tabular approach, if one would be interested instead in the PSPE outcome, i.e. the vectors of game values of for a different initial state (with different heap size or cumulation) the backward induction had to start all over again. The idea of an efficient outcome function is that a computation of game value vectors for a given heap size is universally valid, within the same ruleset, and can be linearly shifted with respect to variations in the initial cumulations. 

%The distinction, following our examples, will be almost self explanatory: cumulation independency is indeed a central property for appealing outcome functions. But we require a solid framework built on the theory of Extensive Form Games. We define an outcome function for the large class of {\sc Cumulative Games} and then we restrict the class somewhat to see when a reasonable generalization of the tabular approach applies. 

We have not proved that {\sc Auction Pebbles} cannot be interpreted as a cumulation independent self-interest ruleset. But our examples seem to point in the direction that this is impossible. This is highly relevant in terms of our main results in Section~\ref{sec:valout}.

\section{Extensive Form Games}\label{sec:extform}
In Section~\ref{sec:cumgame} we defined a generalization of Cumulative Subtraction Games, and we aim to prove that it can capture all Extensive Form Games. To this end, we begin by defining Extensive Form Games in a modular way.

\begin{definition}
An \emph{Extensive Form Game} is a tuple $G=(F,\vec U)$, where $F=(\n,S,T,s_0,\delta,g)$ is the game form with
\begin{enumerate}
\item $\n$ is a set of $n$ players;
	\item $S$ is a finite set of states;
	\item $T\subseteq S$ is a set of terminal states;
	\item $s_0\in S$ is an initial state;
	\item $\delta:S \rightarrow \n$ is a turn function;%\footnote{We assign players to terminal states, even though no action is possible. Intuitively, one reason is that the most important theory in CGT concerns normal play, where by definition a player who cannot move loses. In any CGT situation, the basic terminating issue is: what happens when the current player cannot move? In mis\'ere play they win, and in scoring play a score will be assigned, depending on who is to play. The question in the title ,``who is the current player?'', is fundamental to the study of cumulative/combinatorial games, not just to assign the starting player, but perhaps even more so, to punish or reward a terminal player.}
	\item $g:S\rightarrow 2^S$ is the game function.
	\end{enumerate}
And where $\vec U = (U_1,\ldots,U_n)$, where $U_i:T\rightarrow \mathbb R, i\in \n$ is the utility function of Player~$i$.  
\end{definition}

A note about item 5: we assign players to terminal states, even though no action is possible. Intuitively, one reason is that the most important theory in CGT concerns normal play, where by definition a player who cannot move loses. In any CGT situation, the basic terminating issue is: what happens when the current player cannot move? In mis\'ere play they win, and in scoring play a score will be assigned, depending on who is to play. The question in the title ,``who is the current player?'', is fundamental to the study of cumulative/combinatorial games, not just to assign the starting player, but perhaps even more so, to punish or reward a terminal player.

Intuitively, player~$i = \delta(s)$ is playing in state $s\in S$, and should select the next state from $g(s)\subseteq S$.  We assume that there are no cycles in $g$, and $g(s)=\varnothing$ if and only if $s\in T$. We denote by $S_i\subseteq S$ all states such that $\delta(s)=i$. A state $s'$ is a \emph{descendant} of state $s$ if there is a path in $g$ from $s$ to $s'$. Without loss of generality, we assume that each game state in $S$ is a descendant of $s^0$. 

\paragraph{Strategies and strategy profiles} 
Consider an Extensive Form Game $G=(F,\vec U)$, with $F=(N,S,T,s_0,\delta,g)$.
A Player~$i$ \emph{strategy} is a function $\varsigma_i:S_i\setminus T\rightarrow S$, such that $\varsigma_i(s)\in g(s)$. That is, a unique action (next state) is selected in every state such that $\delta(s)=i$. A \emph{strategy profile} is a vector $\vec\varsigma = (\varsigma_1,\ldots,\varsigma_n)$. Consider any $s\in S$. We denote by $\vec\varsigma|_s$ the restriction of $\vec\varsigma$ on the \emph{subgame} $G|_s$, the subgame of $G$ rooted by $s$.  Denote by $\Sigma(G)$ the set of all strategies in $G$.

\paragraph{Terminal maps and utility maps}
Given a game $G$, and a strategy profile $\vec\varsigma\in \Sigma(G)^n$, the \emph{terminal map} $\tau_{\vec\varsigma}:S\rightarrow T$ maps any state $s$ to the terminal that is reached  when players start from state $s$ and follow their strategies in $\vec\varsigma$ (in the subgame $G|_s$ of $G$). We may omit the parameter $\vec\varsigma$ when  clear from the context. Similarly, the \emph{utility map} $\mu_{i}:\Sigma(G)^n \times S\rightarrow \R$ maps any state to the utility of Player~$i$. That is, $\mu_{i}(\vec\varsigma,s) = U_i( \tau_{\vec\varsigma}(s))$ is the utility to Player~$i$ in game $G|_s$ under strategy profile $\vec\varsigma$. 

\begin{observation}\label{obs:const}
For any profile $\vec\varsigma$, the terminal state $\tau_{\vec\varsigma}(s)$ is constant for any $s$ along the  unique path the profile defines from $s^0$ to $\tau_{\vec\varsigma}(s^0)$, and thus so is $\mu_i(\vec\varsigma,s)$, for all $i\in\n$. 
\end{observation}

\begin{definition}\label{def:PSPE}
Consider a game $G$. A strategy profile $\vec \varsigma=\vec \varsigma^*$ is a \emph{pure subgame perfect Nash equilibrium (PSPE)} if, for all $s\in S\setminus T$, for all $i\in \n$, for any alternative strategy $\varsigma'_i$, 
$$\mu_{i}(\vec\varsigma, s) \ge \mu_{i}((\vec\varsigma_{-i},\varsigma_i'),s).$$
\end{definition}

A game is \emph{generic} if a player is never indifferent between two terminals $\tau$ and $\tau'$ unless $\vec U(\tau)=\vec U(\tau')$. 
Any game can be made generic by specifying some tie-breaking rule (say, lexicographic) in case of indifference.

\begin{property}[Generic Ruleset]
A ruleset is \emph{generic} if each player has a tie breaking rule (preference order of the other players), which is independent of the grounded position. 
\end{property}

Generic games are known to have a unique PSPE utility or \emph{game value} (an $n$-tuple of real values),\footnote{There may be several PSPEs leading to the same value.} which can be found by backward induction. We assume that all our games are generic.
\begin{definition}\label{def:gamevalue}
The game value of an Extensive Form Game $$G=(N,S,T,s_0,\delta,g,\vec{U})$$ is $\ve(G) = (\mu_1(\varsigma^*, s_0), \ldots, \mu_n(\varsigma^*,s_0))$. 
\end{definition}
In the context of grounded {\sc Cumulative Games}, the value $\ve(G)$ is referred to as the \emph{grounded value}. See also Definition~\ref{def:cgvalue} to come.

\section{Cumulative Games' strategy profiles}\label{sec:cumstrat}
A \emph{strategy profile} $\vec\sigma$ for a CGF is an infinite set  of extensive form strategy profiles $\vec\varsigma_{\om,p}$,\footnote{In fact it is an uncountable set since the cumulation vectors are vectors of real numbers.} one for each grounded position $(\om ,p)$, that is \emph{consistent}, in the following sense. For every $(\om',p')$ that is a descendant of $(\om, p)$, the strategy $\vec\varsigma_{\om',p'}$ coincides with $\vec\varsigma_{\om,p}|_{\om',p'}$. Equivalently, for any $\om\in\Omega$, and any $p\in\n$, $\vec\sigma(\om,p)$ selects an action from $\a\in \A(\om)_{p}$.  

\begin{definition}[Cumulative Map]
Consider a CGF, with a given strategy profile $\vec\sigma$. Then $\vec c(\sig,\om,p)=\C(\tau_{\sig}(\om,p))$ is the   \emph{cumulative map} of the grounded position $(\om,p)$. Let $$(\om,p)=(\om^0,p^0),\a^1,(\om^1,p^1),\a^2,\ldots,\a^k,(\om^k,p^k)=\tau_{\sig}(\om)$$ be the sequence of grounded positions and actions from $\om$ in profile $\sig$. Then, for each heap $h$, $$\vec c_{h}(\sig, \om,p)= \C_{h}(\om,p) + \sum_{i= 1}^{k-1} \re_{h}(\om^i,p^i,\a^{i+1},\delta(\om^i,p^i)),$$ and
$\vec c(\sig, \om,p)=\sum_{h\in[d]} \vec c_{h}(\sig, \om,p)$.
\end{definition}

A PSPE in a {\sc Cumulative Game} $(R,\gamma,\om,p,\ut)$ is just a PSPE in the induced Extensive Form Game, $G(R,\gamma, \om,p,\ut)$; we  assume that the game $G$ is generic unless otherwise stated, and thus each player has a unique preference order of the other players. Since this is a property inherited from EG, we strengthen it somewhat by assuring a global uniformity (Layer~1).

\begin{definition}\label{def:cgvalue}
The \emph{grounded game value} of a grounded {\sc Cumulative Game} $(R,\gamma,\om,p,\ut)$ is the game value, $\vec v(G)\in\R^n$, of the induced Extensive Form Game $G$. %$G(R,\gamma,\om,p,\ut)$, $\vec v(G)\in\R^n$.
\end{definition}
Although this definition gives a valid and important notion for a combinatorial game, it is not what is usually called a `game value'. %The reason for this is better understood when we first develop the outcome function, which outputs a vector of grounded game values. Even though this stronger notion, which does not depend on the previous player, it is usually not sufficient to describe (perfect play in) a setting of combinatorial games, namely one of the most important properties is how a combinatorial game behaves together with other games in the same class, using the operation of disjunctive sum. 
For the discussion to proceed, one needs to define the notion of an outcome function. While we do this, we will also prove some important properties of the function, while restricting the class of all {\sc Cumulative Games} appropriately.

\section{Outcome Functions and the Main Results}\label{sec:valout}

We generalize the CGT-outcome function e.g. \cite{BCG} in our setting for {\sc Cumulative Game}s to include general-sum games, and to allow for any prescribed strategy profile (not just the ones in PSPE). Our definition generalizes the outcome functions of Section~\ref{sec:cumsub}. The natural self-interest restriction of the utility function will be required, for the outcome to carry the important PSPE-information (Theorem~\ref{thm:main1}), assuming that games are generic.%The below results hold for the multi heap situation, but, for simplicity, we choose to state and prove everything for the one-heap case, to emphasize the ideas (before the technology). 

In Definition~\ref{def:out}, the notation $\vec\omega^{(\vec a)}$, is the induced $\vec\sigma$-option of $\vec\omega$, but without mention of the reward and  cumulation. This is intentional, as seen in the proof of Lemma~\ref{lem:main}. 

\begin{definition}[$\sig$-outcome]\label{def:out}
Consider a given profile $\vec \sigma$ for a Cumulative Game Form. Let $\vec o_{\vec\sigma}: \Omega\times \n\rightarrow \R^n$. 
For any grounded terminal position $(\vec\omega,p)^\T$, set $\vec o_{\vec\sigma}(\vec\omega, p) := \vec 0$. 
For any non-terminal $(\vec\omega, p)$, with $\vec a = \vec\sigma(\vec\omega)\in \A(\vec\omega)$, set
$$\vec o_{\vec\sigma}(\vec\omega, p) := \vec o_{\vec\sigma}\left(\vec\omega^{(\vec a)},\gamma(\vec\omega,p)\right)+\re\left(\vec\omega, p, \vec a\right),$$ 
The $\vec{\sigma}$-outcome is an $n\times n$-matrix, denoted $\vec o_{\vec\sigma}(\vec\omega)$, one row vector $\vec o_{\vec\sigma}(\vec\omega, p)$ for each current player $p$. The row vectors are the groundeded $\vec\sigma$-outcomes. 
\end{definition}
In traditional CGT, the outcomes describe the result in optimal/perfect play, for either player as a starting player.  Similarly, by Theroem~\ref{thm:main1}, if for each player, the utility function is the identity funtion, the $i$th column vector of $\vec o_{\vec\sigma}(\vec\omega)$, $\vec o^i_{\vec\sigma}(\vec\omega)$, describes Player~$i$'s  result, in terms of accumulated rewards, depending on who starts, given the profile $\vec\sigma$, and if the initial cumulation vector is $\vec 0$. This does not depend on whether the rules are cumulation dependent. Cumulation independency is important if we want a tractable outcome computation, as the second main result, Theorem~\ref{thm:main2} shows. The lemma shows that the rewards act the anticipated way, with respect to the outcome function. The initial cumulation $\C$ may affect the actions and/or rewards, but the outcome function is sensitive to any such influence, so we do not require cumulation independency. The outcome function is general, as it should.

%Note that,

\begin{lemma}\label{lem:main}
Consider a feasible CGF, and a given strategy profile $\vec\sigma$. 
For any grounded position $(\vec\omega,p)$, the grounded \sig-outcome is 
$\vec o_{\vec\sigma}(\vec\omega, p) = \vec c({\vec\sigma}, \vec\omega,p) - \vec C(\vec\omega, p)$.
\end{lemma}
\begin{proof}
Indeed, if $(\vec\omega,p)\in T$, then by definition $\vec c({\vec\sigma},\vec\omega,p) - \vec C(\vec\omega,p)=\vec 0=\vec o_{\vec\sigma}(\vec\omega, p)$. Otherwise, if $(\vec\omega,p)\in (\Omega\times\n)\setminus T$, then
\begin{align}
o_{\vec\sigma}(\vec\omega, p)&= o_{\vec\sigma}\left(\vec\omega^{(\vec a)},\gamma(\vec\omega, p)\right)+\vec r_i\left(\vec\omega, p,\vec a\right)\label{eq:out1}\\
&=\left (\vec c \left({\vec\sigma},\vec\omega^{(\vec a)},\gamma(\vec\omega, p)\right) - \vec C\left(\vec\omega^{(a)},\gamma(\vec\omega, p)\right)\right)+\vec r \left(\vec\omega, p,\vec a\right)\label{eq:out2}\\
&=\vec c({\vec\sigma},\vec\omega,p) - \left(\vec C \left(\vec\omega^{(\vec a)},\gamma(\vec\omega, p)\right)-\vec r \left(\vec\omega, p,\vec a\right)\right)\label{eq:obs}\\
&=\vec c ({\vec\sigma},\vec\omega,p) - \vec C \left(\vec\omega,p\right).\label{eq:out4}
\end{align}
For \eqref{eq:out1}, we use the recursive definition of the outcome function, and for  \eqref{eq:out2}, we use the induction hypothesis. For \eqref{eq:obs}, we use Observation~\ref{obs:const}.  For \eqref{eq:out4}, recall that change in $\vec C$ after action $\vec a$ is exactly the reward.  
\end{proof}

The outcome function is used in zero-sum combinatorial Subtraction Games to characterize the optimal outcomes \cite{BCG,Si}. However it is also meaningful when considering self-interest extensive-form games. Next we show how the \sig-outcome finds the vector of grounded game values in case of generic {\sc Cumulative Games} with self-interest utility.

\begin{theorem}[Outcome-Utility Connection]\label{thm:main1}
Consider a feasible self-interest {\sc Cumulative Game}, and a given strategy profile $\vec\sigma$. 
Then, for any grounded position $(\vec\omega,p)=(\vec x,\C,p)$, Player~$i$'s utility is 
$\mu_i(\vec\sigma,(\vec\omega,p))=o^i_{\vec\sigma}(\vec\omega, p)+C_i(\vec\omega, p).$
In particular, this holds for a strategy profile $\vec\sigma^*$ in PSPE.
\end{theorem}
\begin{proof}
\begin{align*}
\mu_i(\vec\sigma,(\vec\omega,p)) &= c_i({\vec\sigma},\vec\omega,p) \tag{by identity utility}\\
	&= o^i_{\vec\sigma}(\vec\omega,p)+C_i(\vec\omega,p), \tag{by Lemma~\ref{lem:main}}
\end{align*}
as required.
\end{proof}

\begin{corollary} The grounded game value of a grounded position $G=(\vec\omega,p)$ as in Theorem~\ref{thm:main1} is 
 $\vec  v(G)=\vec o_{\vec\sigma^*}(\vec\omega, p)+\C$.
\end{corollary}
\begin{proof}
Obvious.
\end{proof}
Although Theorem~\ref{thm:main1} finds the grounded game values via the outcome function, it is not yet evident whether a reasonable efficient algorithm (type the one-row approach in the tables in Sections~\ref{sec:motivnorm} and \ref{sec:cumsub}) could find these values. However, with appropriate restrictions, we can shed some light on this issue.

Recall Property~\ref{prope:cumind}, cumulation independency of rules. %and Example~\ref{ex:cumrew}%It points at an important property of rewards. %Recall Property~\ref{prope:cumind}. 
%\begin{property}[Cumulation Independent Reward] 
%If the reward does not depend on the cumulation, then the reward is \emph{cumulation independent}.
%\end{property}
A player gets rewarded by the actions they take, possibly depending on the heap size, but not by the history of the game. This notion assumes that actions are cumulation independent.  If the rewards are cumulation independent, we write $\re(\vec\omega,p,\vec a) = \re(\vec x,p,\vec a)$, for $\vec x=\vec x(\vec\omega)$. Similarly, if the actions are independent of the cumulations, we write $\A(\vec \omega,p)=\A(\vec x,p)$. Similarly, the turn function becomes instead $\gamma(\vec x,p)$.

\begin{property}[Short Ruleset]
Consider a feasible ruleset. If, for any $\vec\omega$ and any player $p$, the set $\A(\vec\omega)_p$ (or $\A(\vec x)_p$) is finite, then the ruleset is \emph{short}.
\end{property}

Note, that if the ruleset is short, then the sizes of $\A(\vec x)_p$ may not be bounded in terms of the heap sizes in $\vec x$, but, for all $\vec x$ and all $p$, $|\A(\vec x,p)|<\infty$.
\begin{property}[Heap Dynamic]\label{prope:heapsizedyn}
A short and generic ruleset is \emph{heap dynamic} if all actions, all rewards, and the turn function, are independent of the cumulations in all heap  positions. Sometimes we regard all such games as a ruleset  called {\sc Heap Dynamic}.
\end{property}
Note that the variations of {\sc Cumulative Subtraction} as defined in Section~\ref{sec:cumsub} are heap dynamic. %We state one more (common) restriction of a ruleset. 

We state the recursive computation of the $\sig^*$-outcome for a given generic ruleset with cumulation independent reward. The outcome function becomes particularly simple in the case of Heap Dynamic games, and implies the existence of a dynamic programming algorithm to solve the game.\footnote{Such algorithms can also have theoretical importance, as was recently shown in \cite{CLN}, where dynamic programming approach lead to discovery of a simulation of a two-player normal play game via a one dimensional (diamond shaped) cellular automaton.} 

\begin{theorem}[Efficient Outcome]\label{thm:main2}
Consider an instance of Heap Dynamic, with self-interest utility, and let $\vec o = \vec o_{\vec \sigma^*}$. For any grounded position $(\vec \omega, p)$, if $\A(\vec\omega)_p=\varnothing$, then $\vec o(\vec\omega)=\vec 0$, and otherwise 

\begin{align}\label{eq:outmax}
o^p(\vec\omega,p)=\max_{a\in A(\vec\omega)_p}\{o^p(\vec\omega^{(\vec a)},\gamma(\vec x,p))-r_p(\vec x,p,\vec a)\}, 
\end{align}
and if $i\ne p$, then $o^i(\vec\omega,p)=o^i(\vec\omega^{(a^{\!*}))},\gamma(\vec x,p))-r_i(x,p,a^{\!*})$, where $a^{\!*}$ is a generic maximizing action, i.e. an action that in case of indifference follows the preference order of player~$p$. 

\end{theorem}
\begin{proof}
Combine Definition~\ref{def:util} with Theorem~\ref{thm:main1}. Observe that, since the rewards are cumulation independent, then the $\max$ operator is well defined, and, because the ruleset is short, there are only finitely many actions available, for any given position $\vec\omega$. Indeed, even if the actions are cumulation dependent, we can find a required generic $a^{\!*}$, and similarly for $\gamma(\vec x)$. 
\end{proof}

\begin{observation}[Heap Dynamic Complexity]\label{obs:complexity}
If the {\sc Cumulative Game} $G$ is Heap Dynamic, then the complexity of finding the outcome, the vector of game values is linear in the input size which is bounded by $\max |A(x)|\cdot \text{rank} (G)$. In the case of a subtraction game $G=S(x)$, this is bounded by $x^2$. 
If actions depend on the cumulations, we cannot a priori say anything about the input size in terms of $x$, apart from bounding it by the number of game states, which is in general exponential in rank$(G)$. 
\end{observation}

\section{Partially ordered heap monoids}\label{sec:disum}
 %Here, the notation $\p$, denotes cyclic move order, i.e., for all $p\in \n$, $\p = p \pmod n +1$, which  is indeed the standard in CGT.%\footnote{This is important, for addition of game positions including a state dependent turn function (as in EGT) might be very complicated.}

The real elegance of CGT starts with the notion of disjunctive sum, game comparison and the partial order induced by the outcome function \cite{BCG, Co, Si}. For all this to make sense, the perhaps most important property is that of additive closure. For, if we add two games in a well defined class of games, then their sum should remain in the same class. Traditional recreational rulesets are mostly not additively closed, e.g. Tic-tac-toe: if we add two Tic-tac-toe positions the resulting game does not (easily) correspond to another Tic-tac-toe position. On the other hand, rulesets composed in the context of CGT often satisfy closure properties,\footnote{But not always. If not one can define a ruleset closure of all sums of positions in a given ruleset; se e.g. \cite{PS, Mi}. This is not necessary in our study. Therefore, our class has potentially good properties for future theoretical work generalizing the current framework of CGT to include general-sum, and so on.} such as Nim, Domineering, Hackenbush, and many others. In normal play though, all theory goes through, without investigating additivity properties of specific rulesets. In mis\`ere play, various restrictions are required to analyze disjunctive sums and partial orders.

In this section, we designate the use of letters $G,H$ to Layer~2 games, that is (ungrounded) heap positions of Cumulative Game Forms. 
In Section~\ref{sec:layers}, we defined a disjunctive sum of Layer~2 $n$-player games, $G+_{\! n}H$, but without specifying the details of the addition rule for {\sc Cumulative Games} as defined in Section~\ref{sec:CGF}. The sum of the games belongs to the game space $\Omega(d^+,n)$, where $d^+ = {d_H+d_G}$, the total number of heaps in both $G$ and $H$. 

Let us define vector concatenation. If $\vec x\in \mathbb N^d$ and $\vec y\in \mathbb N^{f}$ then $$\vec x \vec y = (x_1,\ldots x_d,y_1,\ldots , y_{f})\in \mathbb N^{d+f}.$$ 
For all players $i\in\n$, let $$\A^+(G+_nH)_i = \{\vec a\vec 0 \mid \a \in \A(G)_i,\vec 0\in \N^{d_H}\}\cup\{\vec 0\vec a\mid \vec 0\in \N^{d_G}, \a\in\A(H)_i\}.$$ 
For all players $i\in\n$,  for all $\a\in\A(G)_i$, let $$\re^+(G+_n H,i,\a\vec 0) =\re_G(G,i,\a)\, \vec 0,$$
and, for all $\a\in\A(H)_i$, let $$\re^+(G+_n H,i,\vec 0\a) =\vec 0 \, \re_{\!H}(H,i,\a).$$

Thus, altogether, all actions of $G+_nH$ are well defined, and they have well defined rewards. And so the disjunctive sum rulesets are well defined. By combining this with the recursive definition of a disjunctive sum of games \eqref{eq:disj}, we get:

\begin{theorem}
The sum of two heap positions of a Cumulative Game Form is a heap position of a Cumulative game Form. %Namely if $G$
\end{theorem}

Note, that the turn function is irrelevant on Layer~2. But, if we want to compute outcomes, utilities, and so on, then the turn function must be defined, so we suggest a given $n$-player turn function, when we add games.

\begin{definition}\label{def:comp}
Let $G$ and $H$ be $n$-player {\sc Cumulative Games}, on a given turn function. Then player $i$ weakly prefers game $G$ to game $H$, i.e. $G \ge_p H$ if, for all {\sc Cumulative Games} $X$, $\vec o^i(G+_nX)\ge \vec o^i(H+_nX)$. That is, for all starting players $p$, $o^i(G+X,p)\ge o^i(H+X,p)$. 
\end{definition}

\begin{theorem}
The class of {\sc Cumulative Games} is a partially ordered heap monoid, under disjunctive sum, with respect to player $j\in\n$, say. Similarly, the restriction {\sc Heap Dynamic} is a partially ordered heap monoid.
\end{theorem}
\begin{proof}
Definition~\ref{def:comp} is well defined, because of the closure of the disjunctive sum operator: the outcome function accepts {\sc Cumulative Games}. Similarly, the sum of two heap dynamic games belongs to {\sc Heap Dynamic}, and so, if we also restrict the ``for all $X$'' part in Definition~\ref{def:comp}, we have another partial order.
\end{proof}
Similarly, one can have a subclass of {\sc Cumulative Games} where the rewards are cumulation independent (but where actions may depend on cumulations), and this class would again satisfy all closure properties, and therefore define a partial order specific for this class of games. In CGT, usually, when one restricts the class of attention, then the partial order changes. This is mostly studied in the setting of Mis\`ere games (see \cite{MR, Si2} for surveys). Restriction to subclasses of games can be important to obtain efficient reductions of games, to increase the sizes of the equivalence classes of games.

We have the following conjecture.

\begin{conjecture}
All cumulative games are incomparable, and all heap dynamic games are incomparable.
\end{conjecture}

It is probably required to search for more efficiant restrictions, in the pursuit of finding any interesting structure of a heap monoid of {\sc Cumulative Games}. There are sub classes with a lot of structure, since we can model for example normal play, by selecting rewards and utilities as appropriate, and let $n=2$ players.

\section{Strategic Equivalence of Games}\label{sec:equiv}
Given a position $\om^0$ and a previous player $p^0$, we define the set of all possible descendants under a given ruleset $R$ as $$D_R(\om^0,p^0)=\{(\om,p)\mid \text{there is an $R$-path from $(\om^0,p^0)$ to $(\om,p)$}\}$$

For a given ruleset $R$ and initial state $s^0=(\om^0,p^0)$, let $\length(s)$ denote the number of actions from $s^0$ to $s\in D_R(s^0)$.\footnote{We are abusing notation here a bit, the reason being Observation~\ref{obs:cumext} to come. We are identifying states with paths, since they may not be uniquely identified otherwise in a {\sc Cumulative Game}.} Let $\text{rank}(s^0)=\max \{\length{(s)}\mid s\in D_R(\omega^0,p^0), p^0\in N\}$.

\subsection{Every Cumulative Game is an Extensive Form Game}\label{sec:CGEFG}
A ruleset $\Ru$ together with a game position $\om=(\x,\C)$ and a utility function $\ut$ defines an $n$-tuple of Extensive Form Games. Every grounded {\sc Cumulative Game} is an Extensive Form Game.
\begin{observation}\label{obs:cumext} 
 Any grounded {\sc Cumulative Game}, with an initial state $\om^0\in \Omega$ and a previous player $p^0\in \n$, defines a unique extensive form $F=(\n,S,T,s_0,\delta,g)$, where 
\begin{enumerate}
	\item $S=D_R(\om^0,p^0)$;
	\item $T= \{s\in S \mid A(s,\gamma(s))=\varnothing\}$;
	\item $s^0 = (\om^0,p^0)$;
	\item For any $s=(\om,p)\in S$, $\delta(s)=\gamma(\om,p)$;
	\item For any $s=(\om,p) \in S$, $g(s) = \{\left(\om^{(\a)}, \gamma(s)\right) : \a\in A(s,\gamma(s))\}$, where $\om^{(\a)}=\left( \x + \a,\C+\re(\om,p,\a)\right)$. 
	\end{enumerate}
	Moreover, a grounded {\sc Cumulative Game} together with a utility function $\ut$ defines the Extensive Form Game $G=(F,U)$ where
	\begin{enumerate}
	\item[6.] For any $t=(\vec x,\C) \in T$ and $i\in [n]$, $U_i(t) = u_i(\C,\delta(t))$. 
\end{enumerate}
\end{observation}
\begin{definition}
We denote by $G(\Ru,\gamma,\om^0,p^0,\ut)$ the Extensive Form Game induced by the the respective grounded {\sc Cumulative Game}  $(\Ru,\gamma,\om^0,p^0,\ut)$, i.e. $(\om^0,p^0)$ under form $(R,\gamma)$ and utility $\ut$. 
\end{definition}

\subsection{Any Extensive Form Game is a one heap Cumulative Game}
Since the move order is arbitrary in EGT but cyclic in the heap games, we cannot hope to find a heap game for each Extensive Form Game. However, we can do almost as well by introducing equivalence classes of Extensive Form Games. 

Consider an Extensive Form Game $G$. We introduce a reduced form of $G$, $\red(G)$ and a cyclic extension, $\cyc(G)$. The idea is that two games are \emph{strategically equivalent} if we bypass each child with exactly one option. When there is no further bypass possible of $G$, we call this game $\red(G)$. Reversely, we can adjoin a sequence of states (each state with exactly one child) between a parent and a child, to \emph{cycle-complete} the game modulo the $n$ players. For any parent-child $s,s'$ such that $\delta(s') = \delta(s)+ k$, we connect $s,s'$ using a path of $k-1$ states $s_1,\ldots , s_{k-1}$ instead of the single edge $(s,s')$, and where $\delta(s_i)=\delta(s)+i$, for all $i$. When each child is cycle-completed, we call this game $\cyc(G)$. The terminal utilities remain the same, and note that terminal states cannot be bypassed. 

\begin{theorem}\label{thm:exthea}
For each Extensive Form Game $G=(\n,S,T,s_0,\delta,g,U)$ there exists a strategically equivalent grounded {\sc Cumulative Game} with a single heap. 
\end{theorem}
\begin{proof}
Enumerate all states in $S$ from $s_0$ in preorder, so that every node precedes all of its descendents. Denote by $q(s)$ the index of $s$ in this order, and let $Q=|S|$ be the maximal index. Since $q$ is a one-to-one mapping $q:S\rightarrow [Q]$, the function $q^{-1}(z)\in S$ in well defined, and we set $s(x) = q^{-1}(Q-x)$. 

W.l.o.g., $\delta(s')$ is the same for all $s'\in g(s)$ (otherwise we can add dummy states). 

A  position $\omega=(x,\vec C)$ is \emph{valid} if $x+C_j=Q$ for all $j\in N$. Let $\Omega$ contain all valid positions.

Intuitively, every state $s$ in the original game $G$ corresponds to a valid grounded position with heap size $x=Q-q(s)$, and where all players have the same cumulation $C_j = q(s)$. 
We complete the definition of the ruleset $R=([n],d=1,\Omega,A,\vec r)$ as follows. 
Set $A(x) := \{x-y : s(y) \in g(s(x))\}$.
The move order function $\gamma$ is defined as $\gamma(\omega,p):=\delta(s(x))$ where $\omega=(x,\vec C)$.

The (identical reward, identity) reward function is $r(\omega,a,p) := (a,a,\ldots,a)$. This means that if players start from some valid position $\omega = (x,\vec C)$ then after action $a$ that reaches state $\omega = (x',\vec C')$, we will have that $C'_j=(Q-x)+a =(Q-x)+(x-x') = Q-x'$. Thus $\omega'$ is also valid. 

As we intended, every state $s\in S$ in the original game $G$, induces a unique valid position $\omega_s=(x=Q-q(s), \vec C =(q(s),q(s),\ldots,q(s)))$ in the new ruleset $R$.

The initial previous player $p_0$ can be set arbitrarily, as $\gamma$ essentially ignores it.

Finally, we define the self-interest utility function as $u_i(C_i) := U_i(s(C_i))$ if $s(C_i)\in T$ and otherwise $0$. 

We claim that the game $(R,\omega_{s_0},p_0,u)$ is equivalent to $G$. 
 By induction, every move from $s=(\omega,p)\in S$ to $s'=(\omega',p')\in g(s)$ corresponds to a move from $\omega_s$ to $\omega_{s'}$ where $p'=\gamma(\omega,p)$ plays action $a=q(s')-q(s)=x-x'$. When  players in $G$ reach a terminal $t\in T$ and earn $U_i(t)$ each, the corresponding player in $(R,\omega_{s_0},p_0,u)$ gets $u_i(C_i) = U_i(s(C_i))= U_i(s(q(t))) = U_i(t)$, as required.
 \end{proof}

 For example, if $g(s) = \{s',s''\}$ and $q(s)=4, q(s') = 10, q(s'') = 13$, $Q=100$ then $s,s',s''$ correspond to $x=96, x'=90, x''=87$, respectively. $A(96) = \{ 96-90, 96-87\} = \{6,9\}$. Note that $A(x)=\emptyset$ if and only if $g(s(x))=\emptyset$, i.e. if and only if $s(x)\in T$. 

We have another proof of this result, slightly modified, namely where we impose a cyclic turn function on the one heap game, which is the standard in CGT.
\begin{theorem}\label{thm:exthea2}
For each Extensive Form Game $G=(\n,S,T,s_0,\delta,g,U)$ there exists a strategically equivalent grounded {\sc Cumulative Game} with a single heap, and a cyclic turn function. 
\end{theorem}
\begin{proof}
We construct a one heap game $\Delta$, with $x_0 = 0$ and $\C_0=\vec 0$, and we will let the $n$ players increase the heap size, assuming they have sufficient budgets of pebbles. Study $\cyc(G)$. The, say $a$, children of $s_0$, the root of $\cyc(G)$, can be enumerated $s_1,\ldots , s_a$. For each child $s_i$ let the heap size be $x_0+i$, so $A(s_0) = \{1,\ldots , a\}$, and the cumulation is updated trivially (no rewards).  

Study an arbitrary non-terminal node $s$ at $\length(s)=\ell$, and suppose that the heap sizes on level $\ell-1$ range between $X$ and $Y$, with $X<Y$. We enumerate all children at depth $\ell$, denoted say $s_{\ell,1},\ldots s_{\ell,\alpha}$, and let the actions be $Y-X, Y-X+1,\ldots , Y-X+\alpha$. These actions are applied to the heap sizes of the parents at depth $\ell-1$ in increasing order, and we may assume the parents were enumerated in non-decreasing order. Again, the rewards are trivial, unless the action is to a terminal state. In this way we obtain a 1-1 mapping of heap sizes with the original game states in the EG. It remains to assign the correct utilities, and this will be achieved by setting them to the corresponding rewards for each terminating action. Empty sets of actions are attached to the terminal heap sizes.
\end{proof}

Consider a cyclic turn function. By combining Observation~\ref{obs:cumext} with Theorem~\ref{thm:exthea2}, a consequence is that any multi-heap game can be simulated by a single heap game. In fact, this is somewhat simpler than the generic case since a multi-heap game regarded as an Extensive Form Game is already cycle-completed. 
\begin{corollary}
Consider a cyclic turn function. Each grounded multi-heap game has an equivalent one heap game. That is, their game trees and utilities are the same.
\end{corollary}
\begin{proof}
Combine Observation~\ref{obs:cumext} with Theorem~\ref{thm:exthea}.
\end{proof}

Observe that the projection of multiple heaps to one heap here is on layer~3, grounded positions, whereas the famous result that any multi-heap nim position is equivalent to a single heap is a layer~2 result, in the sense that a starting player is not fixed. Equivalence on layer~2, with arbitrary starting player, is in general harder, and we will briefly return to this question in Section~\ref{sec:disum}.

In fact any combinatorial game with cyclic (or alternating) turns, with a given starting player, is equivalent to a one heap {\sc Cumulative Game}, since grounded combinatorial games are Extensive Form Games.

\section{Discussion}\label{sec:discussion}

Let us return to the simpler cases of analysis, concerning zero-sum and self-interest Subtraction Games. 
Any action, that is consistent with an outcome function as in Definitions~\ref{def:outS} and~\ref{def:outCSgs}, will be called an {\em optimal-action}, in a given context.

\begin{observation}[Zero-sum vs. Self-interest]\label{obs:zssi}
One of the first general sum questions for Subtraction Games is: for fixed heap sizes, and fixed subtraction sets, when do Definitions~\ref{def:outS} and~\ref{def:outCSgs} assign the same optimal-action sets? A first (probably correct) guess is that the antagonistic variation is much closer to the zero-sum setting, than the friendly variation. One interesting problem is to explore precisely how much closer that is. And we provide some intuition via some preliminary computations. 

For subtraction sets of size 2, we have not yet detected any difference between zero-sum and self-interest antagonistic optimal-actions.  But for the friendly variation, the first difference appears already on the subtraction set $S=\{3,5\}$, at heap size 14, and where $o_\zs(14)=3$, obtained by subtracting 5 pebbles, but $o_\si^1(14)-o_\si^2(14)=2$, which is obtained by subtracting 3 pebbles. For subtraction sets of size 3, we have detected the first difference of the antagonistic and zero-sum variations for the subtraction set $S=\{6,13,17\}$, at a heap of size 76. The optimal action is  either 6 or 17, and $o_\si^1(76)-o_\si^2(76)=4$, whereas $o_\zs(76)=5$. Moreover, the number of such {\em critical} two or three element subtraction sets with numbers weakly smaller than 20 is one for the antagonistic case, whereas in the friendly case we find altogether 493 cases. If we increase to 30 we find 16 and 2081cases respectively and by increasing  40, we find 68 and 5386 cases respectively. We conjecture that both these numbers grow towards infinity with $\max S$. Fix a  subtraction set. We conjecture that the outcome discrepancy is bounded (for either antagonistic or friendly tie breaking convention), and this would be a corollary of another conjecture, that the greedy action is eventually an optimal action for self-interest (independent of tie-breaking convention); this was proved for zero-sum games in \cite{CLMW}. On the other hand, we conjecture that the outcome discrepancy can be arbitrarily large when we let the subtraction set vary.  
\end{observation}

Let us mention some examples where Pareto efficiency is the interesting concept.
\begin{example}[Tragedy of the Common]\label{ex:pareto}
In welfare economics, for general sum games, at the core of the heart is the notion of Pareto efficiency (PE). A Pareto efficient play sequence is such that it is impossible to reallocate the actions so as to make any one player better off without making another one worse off. Here, we need to respect that the number of actions for the starting (earlier) player is either the same as the other players or they have one more action.

Consider the two-player symmetric CS game $S = \{20, 31, 51\}$, playing from $x = 100$, with identity rewards and identity utilities. The unique play sequence in PSPE is for Player~1 to take 51, and then Player~2 takes 31. At this point, no further action is possible. Any other play from Player~1 would give Player~2 the opportunity to take 51, and so she would get at most 40. However, there is a solution, which is better for both players. It is when both players have agreed beforehand to take 20 in each move. Then all resources will be allocated, which implies Pareto efficiency. 

Returning to the main example with $S = \{2,3\}$ with identity rewards and identity utilities. Then, playing from any heap size it is easy to see that the outcome is Pareto efficient, both for antagonistic and cooperative tie break rule.

The smallest symmetric subtraction game that is not Pareto efficient for many heap sizes is $S=\{3,7\}$, and the first starting position that fails is $x=30$. The reason is that Player~1 cannot afford to make a big sacrifice and play 3. Because then Player~2 easily responds with 7, and now Player~1 starts from a heap of size 20, which has game value $(10,10)$. Thus the utility would be 13, when they obtain 14 by playing 7 twice. But clearly a Pareto efficient cooperation would yield the utilities $(15,15)$. 
\end{example}
We propose a resolution to the tragedy of the common in the setting of {\sc Cumulative Game}s. By introducing a principal-agent situation, one can offer a solution to the tragedy of the common in Example~\ref{ex:pareto}, by letting the principal choose the reward functions, and the agents play as usual, given the suggested reward function by the principal. So, when they play they seek to maximize the utility as defined here. However, there is a second stage of game, when the principal reveals their `true utilities', and in the example they are simply their respective sum of their actions: set the rewards for each action to $r(s)=-s$. Then indeed the Pareto (efficient) solution is achieved. In this way, rewards can be seen as a way to correct incitements that go wrong because of perhaps too individualistic behavior.

\paragraph{Memory.} A big class of normal play rulesets have `memory', meaning that the action set may depend on previous actions (in full generality the action sets may depend on all history of a game).  The most famous such ruleset in the normal play theory is {\sc Fibonacci Nim} \cite{W}, which depends only on the most recent action, and it has recently been generalized to multi heap situations \cite{LR-S,LR-S2}. Another game in this family that has deeper memory is {\sc Imitation Nim} \cite{L2}. For a theoretical purpose, such history dependencies are usually simply included into the notion of a game position, which permits all standard CGT-tools.

Note that the cumulations in this paper erase any particular memory of ``the previous move'' or similar. 
Here, any notion of memory (including possible dependencies) is stored in a cumulation vector (which is included into the notion of a heap position). In theory we could have memory be more general, but the point is to have some natural restrictions. We feel that, for an economic type game, the available actions may depend on the players' cumulations/budgets/endowments etc, but `how' they got to their present cumulations can be ignored.

The results in this paper would still hold, but with more cumbersome notation, so for now we omit this class of games.

\paragraph{Cyclic games.} A natural extension of the current work is to study {\sc Cumulative Game}s with cycles. One can define various results `in the limit', using standard $\limsup$ and $\liminf$ maps of results along strategy profiles, and thus generalize the outcome functions. 

\paragraph{Philosophy.} We mentioned one difference between EG and CG: it is given who starts in EG but not in CG. Another major distinction concerns the movability, and in particular `who gets to make the final move', often it is useful to have more moves than your opponent, and this is the de facto standard that all CGT should be put in relation with. Extensive Form Games are often more utility oriented, and the notion of who moves last is rarely the most important issue. In fact, even for sequential games a standard assumption is that each player declares their strategies, and so the movability is rarely mentioned. It suffices to know the players' strategy profile, and then their utilities follow. So the underlying philosophy used to be quite different. Here we show that they need not be that different. For example, we can include the movability in various components as rewards and utilities, so normal play could enter EGT as a natural component: how much does a player evaluate a quick gain in a gray zone of the law, compared to the risk of being convicted, and lose the freedome to move?

\paragraph{Can we have non-trivial game comparison for self-interest {\sc Cumulative Game}s?} If one aspires big equivalence classes of games, in the spirit of normal play theory, then it would have to be defined with respect to a given player $p$, as we do in Definition~\ref{def:comp}. For this to be interesting, one would most likely have to restrict the class of games, and the first step would be to find a class of games such that the equivalence class of $\vec 0$ is non-trivial. Indeed, a first step in this direction has been taken in that normal play CGT is bridged with scoring-play CGT in a useful restriction of the full class of scoring games \cite{St} to the class of guaranteed scoring \cite{LNS,LNNS} in which normal play is order embedded; the full class has only the trivial neutral element, whereas for guaranteed scoring the equivalence class of neutral elements contain all 0s of Normal play. This recent  scoring-play development started with Ettinger's seminal Ph.D thesis \cite{E}; he extended the Milnor-type positional (nonnegative incentive scoring-play) games to include so-called zugzwangs (where no player wants to start), a common concept in recreational play.

\paragraph{Solve Wealth Pebbles.} {\sc Wealth Pebbles} P6 is a partizan normal play two-player game, and it is {\em all small} (if one player has a move, then the other player has a move) if we assume that all heaps start with cumulations at least $(1,1)$ (and otherwise it is trivial). Hence, in the context of a disjunctive sum of games, the notion of atomic weights \cite{BCG, Si} will readily arrive as a tool. Thus, a first step to characterize these games would be to determine heap sizes atomic weights as a function of the players current cumulations. Atomic weights are a rough measure of how many times you can afford to `pass' (or wait) in any given component. Of course, if you lead by a certain amount in one component, then a move there is not urgent, but it could be more useful to accumulate more wealth in a neighboring game component. We propose this game among others as a step forward to develop an economic branch of classical CGT. A recent preprint \cite{Ankita} studies a variation of P6 with arbitrary hot positions, called {\sc Robin Hood}.

\end{document}